\documentclass[11pt, letterpaper]{article}

\usepackage[table,svgnames]{xcolor}
\usepackage[colorlinks=true,linkcolor=black,citecolor=MidnightBlue,urlcolor=MidnightBlue]{hyperref}

\usepackage{amssymb,amsmath,amsfonts,fullpage,mdwlist,enumerate,amsthm, graphicx,algpseudocode, algorithm, algorithmicx, cite, enumerate}
\usepackage[letterpaper, margin=1in]{geometry}
\usepackage{thmtools, thm-restate}

\newcommand{\remove}[1]{}
\newcommand{\unTAP}{$A_{TAP}$}
\newcommand{\unAug}{$A_{Aug}$}
\newcommand{\weTAP}{$A_{wTAP}$}
\newcommand{\unAugTag}{$A'_{Aug}$}

\usepackage[normalem]{ulem}

\newtheorem{theorem}{Theorem}[section]
\newtheorem{lemma}[theorem]{Lemma}

\newtheorem{corollary}[theorem]{Corollary}

\newtheorem{claim}[theorem]{Claim}
\newtheorem{definition}{Definition}[section]

\remove{
\newenvironment{theorem-repeat}[1]{\begin{trivlist}
\item[\hspace{\labelsep}{\textcolor{darkgray}{$\blacktriangleright$}\nobreakspace\sffamily\bfseries Theorem \ref{#1}.}]\em }%
{\end{trivlist}}
}

\title{Fast Distributed Approximation for TAP \\and 2-Edge-Connectivity\thanks{A preliminary version of this paper appeared in OPODIS 2017.}}

\author{Keren Censor-Hillel\footnote{Technion, Department of Computer Science, \texttt{\{ckeren,smichald\}@cs.technion.ac.il}. Supported in part by the Israel Science Foundation (grant 1696/14).}
\and Michal Dory\footnotemark[2]{}
}

\begin{document}

\begin{titlepage}

\maketitle

\begin{abstract}
The \emph{tree augmentation problem (TAP)} is a fundamental network design
problem, in which the input is a graph $G$ and a spanning tree $T$ for it,
and the goal is to augment $T$ with a minimum set of edges $Aug$ from $G$,
such that $T \cup Aug$ is 2-edge-connected.

TAP has been widely studied in the sequential setting. The best known
approximation ratio of 2 for the weighted case dates back to the work of 
Frederickson and J{\'{a}}J{\'{a}}, SICOMP 1981. 
Recently, a 3/2-approximation was 
given for unweighted TAP by Kortsarz and Nutov, TALG 2016.
Recent breakthroughs give an approximation of 1.458 for unweighted TAP 
[Grandoni et al., STOC 2018], and approximations better than 2 for bounded
weights [Adjiashvili, SODA 2017; Fiorini et al., SODA 2018].

In this paper, we provide the first fast \emph{distributed} approximations
for TAP. We present a distributed $2$-approximation for weighted TAP which
completes in $O(h)$
rounds, where $h$ is the height of $T$. When $h$ is large, we show a much
faster 4-approximation algorithm for the unweighted case, completing in
$O(D+\sqrt{n}\log^*{n})$ rounds, where $n$ is the number of vertices and
$D$ is the diameter of $G$.

Immediate consequences of our results are an $O(D)$-round 2-approximation
algorithm for the minimum size 2-edge-connected spanning subgraph, which 
significantly improves upon the running time of previous approximation 
algorithms, and an $O(h_{MST}+\sqrt{n}\log^{*}{n})$-round 3-approximation 
algorithm for the weighted case, where $h_{MST}$ is the height of the MST 
of the graph. Additional applications are algorithms for verifying 
2-edge-connectivity and for augmenting the connectivity of any connected 
spanning subgraph to 2.

Finally, we complement our study with proving lower bounds for distributed
approximations of TAP.

\end{abstract}

\thispagestyle{empty} 
\end{titlepage}

\section{Introduction}

The tree augmentation problem (TAP) is a central problem in network design. In TAP, the input is a 2-edge-connected\footnote{A graph $G$ is 2-edge-connected if it remains connected after the removal of any single edge.} graph $G$ and a spanning tree $T$ of $G$, and the goal is to augment $T$ to be 2-edge-connected by adding to it a minimum size (or a minimum weight) set of edges from $G$. Augmenting the connectivity of $T$ makes it resistant to any single link failure, which is crucial for network reliability. TAP is extensively studied in the sequential setting, with several classical 2-approximation algorithms \cite{frederickson1981approximation,khuller1993approximation,goemans1994improved,jain2001factor}, as well as recent advances with the aim of achieving better approximation factors \cite{kortsarz2016simplified, DBLP:journals/corr/FioriniGKS17, adjiashvili2017beating, cheriyan2015approximating, DBLP:conf/stoc/0001KZ18}. 

TAP is part of a wider family of \emph{connectivity augmentation} problems.
Finding a minimum spanning tree (MST) is another prime example for a problem in this family, but, 
although an MST is a low-cost backbone of the graph, it cannot survive even one link failure. Hence, in order to guarantee stronger reliability, it is vital to find subgraphs with higher connectivity. The motivation for considering TAP is for the case that adding any new edge to the backbone incurs a cost, and hence if we are already given a subgraph with some connectivity guarantee then we would naturally like to augment it with additional edges of minimum number or weight, rather than to compute a well-connected low-cost subgraph from scratch. 
Connectivity augmentation problems also serve as building blocks in other connectivity problems, such as computing the minimum $k$-edge-connected subgraph. A natural approach is to start with building a subgraph that satisfies some connectivity guarantee (e.g., a spanning tree), and then augment it to have stronger connectivity. 

Since the main motivation for TAP is improving the reliability of distributed networks, it is vital to consider TAP also from the distributed perspective. 
In this paper, we initiate the study of distributed connectivity augmentation and present the first distributed approximation algorithms for TAP. We do so in the CONGEST model~\cite{peleg2000distributed}, in which vertices exchange messages of $O(\log{n})$ bits in synchronous rounds, where we show fast algorithms for both the unweighted and weighted variants of the problem.
In addition to fast approximations for TAP, our algorithms have the crucial implication of providing efficient algorithms for approximating the minimum 2-edge-connected spanning subgraph, as well as  
several related problems, such as verifying 2-edge-connectivity and augmenting the connectivity of any spanning connected subgraph to 2. 
Finally, we complement our study with proving lower bounds for distributed approximations of TAP.

\subsection{Our Contributions}

\subsubsection*{Distributed approximation algorithms for TAP}

Our first main contribution is the first distributed approximation algorithm for TAP.  In particular, our algorithm provides a 2-approximation for weighted TAP in the CONGEST model, summarized as follows.

\begin{restatable}{theorem}{wTAP} \label{wTAP}
There is a distributed 2-approximation algorithm for weighted TAP in the CONGEST model that runs in $O(h)$ rounds, where $h$ is the height of the tree $T$. 
\end{restatable}

The approximation ratio of our algorithm matches the best approximation ratio for weighted TAP in the sequential setting. Its round complexity of $O(h)$ is tight if $h = O(D)$, where $D$ is the diameter of $G$. This happens, for example, when $T$ is a BFS tree, and follows from a lower bound of $\Omega(D)$ rounds which we show in Section~\ref{sec:lower}. 

However, the height $h$ of the spanning tree $T$ may be large, even if the diameter of $G$ is small, which raises the question of whether the dependence on $h$ is necessary. We address this question by providing an algorithm for \emph{unweighted} TAP that has a round complexity of $O(D+\sqrt{n}\log^*{n})$ rounds, which is significantly smaller for large values of $h$. This only comes at the price of a slight increase in the approximation ratio, from $2$ to $4$. 

\begin{restatable}{theorem}{uTAPtwo} \label{uTAPtwo}
There is a distributed 4-approximation algorithm for unweighted TAP in the CONGEST model that runs in $O(D+\sqrt{n}\log^*{n})$ rounds.
\end{restatable}

\subsubsection*{Applications}

The key application of our TAP approximation algorithm is an $O(D)$-round 2-approximation algorithm for the minimum size 2-edge-connected spanning subgraph problem (2-ECSS), which is obtained by building a BFS tree and augmenting it to a 2-edge-connected subgraph using our algorithm.

\begin{restatable}{theorem}{ECSS}
There is a distributed 2-approximation algorithm for unweighted 2-ECSS in the CONGEST model that completes in $O(D)$ rounds. 
\end{restatable}

The time complexity of our algorithm improves significantly upon the time complexity of previous approximation algorithms for 2-ECSS, which are $O(n)$ rounds for a $\frac{3}{2}$-approximation \cite{krumke2007distributed} and $O(D+\sqrt{n}\log^*{n})$ rounds for a 2-approximation \cite{thurimella1995sub}. 

In addition, our weighted TAP algorithm implies a 3-approximation for \emph{weighted} 2-ECSS.
Other applications of our algorithms are an $O(D)$-round algorithm for verifying 2-edge-connectivity, and an algorithm for augmenting the connectivity of any connected spanning subgraph $H$ of $G$ from $1$ to $2$. 

\subsubsection*{Lower bounds}

We complement our algorithms by presenting lower bounds for TAP. We first show that approximating TAP is a global problem which requires $\Omega(D)$ rounds even in the LOCAL model\cite{Linial92}, where the size of messages is not bounded.

\begin{restatable}{theorem}{local}
\label{local-lb}
Any distributed $\alpha$-approximation algorithm for weighted TAP takes $\Omega(D)$ rounds in the LOCAL model, where $\alpha \geq 1$ can be any polynomial function of $n$. This holds also for unweighted TAP, if $1 \leq \alpha < \frac{n-1}{2c}$ for a constant $c>1$.
\end{restatable}

Theorem~\ref{local-lb} implies that if $h=O(D)$ then our TAP approximation algorithms have an optimal round complexity. We also consider the case of $h = \omega(D)$ and show a family of graphs, based on the construction in \cite{sarma2012distributed}, for which $\Omega(h)$ rounds are needed in order to approximate weighted TAP, were $h=\Theta(\frac{\sqrt{n}}{\log{n}})$. 

\begin{restatable}{theorem}{congest}
\label{theorem:congest-lb}
For any polynomial function $\alpha(n)$, there is a $\Theta(n)$-vertex graph of diameter $\Theta(\log{n})$ for which any (even randomized) distributed $\alpha(n)$-approximation algorithm for weighted TAP with an instance tree $T \subseteq G$ of height $h=\Theta(\frac{\sqrt{n}}{\log{n}})$ requires $\Omega(h)$ rounds in the CONGEST model. 
\end{restatable}

Theorem~\ref{theorem:congest-lb} implies that our algorithm for weighted TAP is optimal on these graphs. In particular, there cannot be an algorithm with a complexity of $O(f(h))$ for a sublinear function $f$. 
This lower bound can also be seen as an $\widetilde{\Omega}(D+\sqrt{n})$ lower bound.

Our lower bound for weighted TAP implies a lower bound for weighted 2-ECSS, since an $\alpha$-approximation algorithm for weighted 2-ECSS gives an
$\alpha$-approximation algorithm for weighted TAP where we give to the edges of the input tree $T$ weight 0.

\subsection{Technical overview of our algorithms}

As an introduction, we start by showing an $O(h)$-round 2-approximation algorithm for \emph{unweighted} TAP, which allows us to present some of the key ingredients in our algorithms. Later, we explain how we build on these ideas and extend them to give an algorithm for the weighted case, and a faster algorithm for unweighted TAP.

\subsubsection*{Unweighted TAP} 

A natural approach for constructing a distributed algorithm for unweighted TAP could be to try to simulate the sequential $2$-approximation algorithm of Khuller and Thurimella \cite{khuller1993approximation}. In their algorithm, the input graph $G$ is first converted into a modified graph $G'$. 
Then, the algorithm finds a directed MST\footnote{A directed spanning tree of $G$ rooted at $r$, is a subgraph $T$ of $G$ such that the undirected version of $T$ is a tree and $T$ contains a directed path from $r$ to any other vertex in $V$. A directed MST is a directed spanning tree of minimum weight.} in $G'$, which induces a corresponding augmentation in $G$. 

When considered in the distributed setting, this approach imposes two difficulties. The first is that we cannot simply modify the input graph, because it is the graph that represents the underlying distributed network, whose topology is given and not under our control. The second
is in the directed MST procedure, as finding a directed MST efficiently in the distributed setting seems to be difficult. The currently best known time complexity of this problem is $O(n^2)$ for an asynchronous setting\cite{humblet1983distributed}, which is trivial in the CONGEST model. 

We overcome the above using two key ingredients. First, we bring into our construction the tool of computing lowest common ancestors (LCAs).
We show that building $G'$ and simulating a distributed computation over it can be done by an efficient computation of LCAs, and we achieve the latter by leveraging the \textit{labeling scheme} for LCAs presented in \cite{alstrup2004nearest}. 

Second, we drastically diverge from the Khuller-Thurimella framework by replacing the expensive directed MST construction by a completely different procedure. Roughly speaking, we show that the simple structure of $G'$ allows us to find an optimal augmentation in $G'$ efficiently by scanning the input tree $T$ from the leaves to the root and performing the following procedure. Each vertex sends to its parent information about edges that may be useful for the augmentation since they \emph{cover} many edges of the tree, and the vertices use the LCA labels in order to decide which edges to add to the augmentation. 

While a direct implementation of this would result in much information that is sent through the tree, we show that at most two edges need to actually be sent by each vertex. Thus, applying the labeling scheme and scanning the tree $T$ result in a time complexity of $O(h)$ rounds, where $h$ is the height of $T$.
Finally, we prove that an optimal augmentation in $G'$ gives a 2-approximation augmentation for $G$, which gives a 2-approximation for unweighted TAP in $O(h)$ rounds. 

\subsubsection*{Weighted TAP} 

Our algorithm for the unweighted case relies heavily on the fact that we can compare edges and decide which one is the best for the augmentation according to the number of edges they cover in the tree. However, once the edges have weights, it is not clear how to compare edges. This is because of the tension between light edges that cover only few edges and heavier edges that cover many edges.
Therefore, Theorem~\ref{wTAP}, which applies for the weighted case, cannot be directly obtained according to the above description. 

Nevertheless, we show how to overcome this by introducing a technique of having each vertex send to its parent edges with \emph{altered weights}. The trick here is that we modify the weight that is sent for an edge in a way that captures the cost for covering each edge of the tree. This successfully addresses the competing needs of covering as many tree edges as possible, while using the lightest possible edges, and allows focusing on a smaller number of edges that may be useful for the augmentation. 
Finally, using standard pipelining, this gives a time complexity of $O(h)$ rounds for the weighted case as well.

\subsubsection*{Faster unweighted TAP} 

Both of our aforementioned algorithms rely on scanning the tree $T$, which results in a time complexity that is linear in the height $h$ of the tree $T$. In order to avoid the dependence on $h$, one must be a able to add edges to the augmentation without scanning the whole tree. 

However, if a vertex $v$ does not get information about the edges added to the augmentation by the vertices in the whole subtree rooted at $v$, then it may add additional edges in order to cover tree edges that are already covered. But then we are no longer guaranteed to get an optimal augmentation in $G'$, or even a good approximation for it.

Nevertheless, we are still able to show a faster algorithm for unweighted TAP, which completes in $O(D+\sqrt{n}\log^*{n})$ rounds. 
The key ingredient in our algorithm is breaking the tree $T$ into fragments and applying our $2$-approximation for unweighted TAP algorithm on each fragment separately, as well as on the tree of fragments. 
Since our algorithm does not scan the whole tree, it may add different edges to cover the same tree edges, which makes the analysis much more involved. The approximation ratio analysis is based on dividing the edges to different types and bounding the number of edges of each type separately, using a subtle case-analysis. 
Although our algorithm does not find an optimal augmentation in $G'$, it gives a 2-approximation for it, which results in a 4-approximation augmentation for the original graph $G$. \\

\textbf{Roadmap:} In Section \ref{sec:app_uTAP}, we describe our $O(h)$-round 2-approximation algorithm for unweighted TAP, and in Section \ref{sec:app_wTAP} we extend it to the weighted case. In Section \ref{sec:applic}, we show applications of these algorithms, in particular for approximating 2-ECSS, and in Section \ref{sec:faster} we present our faster algorithm for unweighted TAP. We present lower bounds for TAP in Section \ref{sec:lower}, and discuss questions for future research in Section \ref{sec:dis}.  

\subsection{Related Work}

\subsubsection*{Sequential algorithms for TAP}

TAP is intensively studied in the sequential setting. Since TAP is NP-hard, approximation algorithms for it have been studied. The first 2-approximation algorithm for weighted TAP was given by Frederickson and J{\'{a}}J{\'{a}} \cite{frederickson1981approximation}, and  was later simplified by Khuller and Thurimella \cite{khuller1993approximation}.
Other 2-approximation algorithms for weighted TAP are the primal-dual algorithm of Goemans et al. \cite{goemans1994improved}, and the iterative rounding algorithm of Jain \cite{jain2001factor}. 

Recently, a new algorithm achieved an approximation of 1.5 for unweighted TAP \cite{kortsarz2016simplified}, and
recent breakthroughs give 1.458-approximation for unweighted TAP \cite{DBLP:conf/stoc/0001KZ18}, and approximations better than 2 for bounded weights \cite{DBLP:journals/corr/FioriniGKS17, adjiashvili2017beating}.
Achieving approximation better than 2 for the general weighted case is a central open question. 
See \cite{khuller1996approximation, kortsarz2010approximating} for surveys about approximation algorithms for connectivity problems. Also, the related work in \cite{DBLP:conf/stoc/0001KZ18} gives an overview of many recent sequential algorithms for TAP. 

\subsubsection*{Related work in the distributed setting}

While ours are the first distributed approximation algorithms for TAP itself, there are important related studies in the distributed setting.

~\\
\textbf{MST:} In the distributed setting, finding an MST, which is a minimum weight subgraph with connectivity $1$, is a fundamental and well studied problem (see, e.g., \cite{gallager1983distributed,garay1998sublinear,kutten1995fast,elkin2006unconditional, DBLP:conf/podc/Elkin17, pandurangan2017time}). 
The first distributed algorithm for this problem is the GHS algorithm that works in $O(n\log{n})$ time \cite{gallager1983distributed}. Following algorithms improved the round complexity to $O(D+\sqrt{n}\log^*{n})$  \cite{garay1998sublinear, kutten1995fast}. 

~\\
\textbf{$k$-ECSS:} For the minimum weight 2-edge-connected spanning subgraph (2-ECSS) problem, there is a distributed algorithm of Krumke et al. \cite{krumke2007distributed}. Their approach is finding a specific spanning tree and then augmenting it to a 2-edge-connected graph. In the unweighted case, they augment a DFS tree following the sequential algorithm of Khuller and Vishkin \cite{khuller1994biconnectivity}, which results in an $O(n)$-round $\frac{3}{2}$-approximation algorithm for 2-ECSS. In the weighted case they augment an MST and suggest a general $O(n\log{n})$-round 2-approximation algorithm for weighted TAP, which gives an $O(n\log{n})$-round $3$-approximation algorithm for 2-ECSS. Our algorithms for TAP imply faster approximations for unweighted and weighted 2-ECSS.

Another distributed algorithm for \emph{unweighted} $k$-ECSS is an $O(k(D+\sqrt{n}\log^*{n}))$-round algorithm of Thurimella \cite{thurimella1995sub} that finds a sparse $k$-edge-connected subgraph. The general framework of the algorithm is to repeatedly find maximal spanning forests in the graph and remove their edges from the graph (this framework is also described in sequential algorithms \cite{khuller1996approximation,nagamochi1992linear}). This gives a $k$-edge-connected spanning subgraph with at most $k(n-1)$ edges. Since any $k$-edge-connected subgraph has at least $\frac{kn}{2}$ edges, since the degree of each vertex is at least $k$, this approach guarantees a 2-approximation for unweighted $k$-ECSS. 

~\\
\textbf{Fault-tolerant tree structures:} Another related problem is the construction of fault-tolerant tree structures. Distributed algorithms for constructing fault tolerant BFS and MST structures are given in \cite{ghaffari2016near}, producing sparse subgraphs of the input graph $G$ that contain a BFS (or an MST) of $G \setminus\{e\}$ for each edge $e$, for the purpose of maintaining the functionality of a BFS (or an MST) even when an edge fails. However, TAP is different from these problems in several aspects. First, we augment a specific spanning tree $T$ rather then build the whole structure from scratch. In addition, since we need to preserve only connectivity when an edge fails and not the functionality of a BFS or an MST, optimal solutions for TAP may be much cheaper. 

~\\
\textbf{Additional related problems:} Another connectivity augmentation problem studied in the distributed setting is the Steiner Forest problem \cite{lenzen2014improved, khan2012efficient}. There are also distributed algorithms for finding the 2-edge-connected and 3-edge-connected components of a connected graph \cite{pritchard2005robust,pritchard2011fast}, and distributed algorithms that decompose a graph with large connectivity into many disjoint trees, while almost preserving the total connectivity through the trees \cite{censor2014distributed}.

\subsubsection*{Follow-up works}

We show here a deterministic $O(D+\sqrt{n}\log^*{n})$-round 4-approximation algorithm for \emph{unweighted} TAP and a determinstic $O(h)$-round 2-approximation algorithm for \emph{weighted} TAP. In a recent follow-up work \cite{kECSS} we show a randomized $O((D+\sqrt{n})\log^2{n})$-round $O(\log{n})$-approximation for \emph{weighted} TAP and \emph{weighted} 2-ECSS, based on different techniques. In addition, we show in \cite{kECSS} a randomized $\widetilde{O}(n)$-round $O(\log{n})$-approximation for \emph{weighted} $k$-ECSS for any constant $k$, and a randomized $O(D\log^3{n})$-round $O(\log{n})$-approximation for \emph{unweighted} 3-ECSS.

Also, a very recent work \cite{2ECSS_new} shows a deterministic $O(1)$-approximation for weighted TAP and weighted 2-ECSS, completing in $O((D+\sqrt{n})\log^2{n})$ rounds. Another very recent work \cite{un_kECSS} shows an $O(1)$-approximation for \emph{unweighted} $k$-ECSS completing in $O(k \log^{1+o(1)}{n})$ rounds. The basic approach in \cite{un_kECSS} is building $k$ ultra-sparse spanners iteratively. Since any ultra-sparse spanner has $O(n)$ edges, the total number of edges in the subgraph obtained is $O(kn)$, which gives a constant approximation for unweighted $k$-ECSS. While these recent works improve significantly the time complexity for weighted TAP and 2-ECSS, and unweighted $2$-ECSS, this comes at a price of larger approximation ratios than the ones we show here. For a detailed comparison see Table \ref{table_results}.

\begin{table}[h!]
\centering
\begin{tabular}{ |p{3.5cm}|p{2cm}|p{3cm}|p{3.5cm}|  }
 \hline
 \multicolumn{4}{|c|}{Algorithms and lower bounds for TAP} \\
 \hline
 Reference& Variant & Approximation & Time complexity\\
 \hline
 \textbf{This paper}   & weighted    &2 &  $O(h)$\\
 \textbf{This paper}   & unweighted    &4 &  $O(D+\sqrt{n}\log^*{n})$\\
 \textbf{This paper}   & unweighted    &$\alpha = O(n)$ &  $\Omega(D)$\\
 \textbf{This paper}   & weighted    &any polynomial $\alpha$&  $\widetilde{\Omega}(D+\sqrt{n}), \Omega(h)$\\
 Subsequent work \cite{kECSS}  & weighted    &$O(\log{n})$ &  $O((D+\sqrt{n})\log^2{n})$\\
 Subsequent work \cite{2ECSS_new}  & weighted    &$O(1)$ &  $O((D+\sqrt{n})\log^2{n})$\\
 \hline
 \multicolumn{4}{|c|}{} \\
 \hline
 \multicolumn{4}{|c|}{Algorithms and lower bounds for weighted 2-ECSS} \\
 \hline
 Reference& Variant & Approximation & Time complexity\\
 \hline
  Prior work \cite{krumke2007distributed}  &    &3 &  $O(n\log{n})$\\
  \textbf{This paper}   &    &3 &  $O(h_{MST} + \sqrt{n} \log^{*}{n})$\\
  \textbf{This paper}   &     &any polynomial $\alpha$&  $\widetilde{\Omega}(D+\sqrt{n})$\\
 Subsequent work \cite{kECSS}  &   &$O(\log{n})$ &  $O((D+\sqrt{n})\log^2{n})$\\
 Subsequent work \cite{2ECSS_new}  &    &$O(1)$ &  $O((D+\sqrt{n})\log^2{n})$\\
 \hline
  \multicolumn{4}{|c|}{} \\
 \hline
 \multicolumn{4}{|c|}{Algorithms for unweighted $k$-ECSS} \\
 \hline
 Reference& Variant & Approximation & Time complexity\\
 \hline
  Prior work \cite{krumke2007distributed}  &  $k=2$  &3/2 &  $O(n)$\\
  Prior work \cite{thurimella1995sub}  &  general $k$  &2 &  $O(k(D+\sqrt{n}\log^*{n}))$\\
  \textbf{This paper}   & $k=2$   &2 &  $O(D)$\\
 Subsequent work \cite{un_kECSS}  & general $k$ &$O(1)$ &  $O(k \log^{1+o(1)}{n})$\\
 \hline
\end{tabular}
 \caption{Summary and comparison of our results}
\label{table_results}
\end{table}

\subsection{Preliminaries}
For completeness, we first formally define the notion of edge connectivity.

\begin{definition}
An undirected graph $G$ is \emph{$k$-edge-connected} if it remains connected after the removal of any $k-1$ edges. 
\end{definition}

\textbf{The Tree Augmentation Problem (TAP).} In TAP, the input is an undirected 2-edge-connected graph $G$ with $n$ vertices, and a spanning tree $T$ of $G$. The goal is to add to $T$ a minimum size (or a minimum weight) set of edges $Aug$ from $G$, such that $T \cup Aug$ is 2-edge-connected. In the weighted version, each edge has a non-negative weight, and we assume that the weights of the edges can be represented in $O(\log n)$ bits. 

\begin{definition}
An edge $e$ in a connected graph $G$ is a \emph{bridge} in $G$ if $G \setminus \{e\}$ is disconnected.
\end{definition}
\begin{definition}
A non-tree edge $e=\{u,v\}$ \emph{covers} the tree edge $e'$ if $e'$ is on the unique path in $T$ between $u$ and $v$, i.e., if $e'$ is not a bridge in $T\cup \{e\}$.
\end{definition}

A graph $G$ is 2-edge-connected if and only if it does not contain bridges. Hence, augmenting the connectivity of $T$ requires covering all the tree edges.

~\\
\textbf{Models of distributed computation.} In the distributed CONGEST model \cite{peleg2000distributed}, the network is modeled as an undirected connected graph $G=(V,E)$. Communication takes place in synchronous rounds. In each round, each vertex can send a message  of $O(\log{n})$ bits to each of its neighbors. The time complexity of an algorithm is measured by the number of rounds. 
Our algorithms work in the CONGEST model, but some of our lower bounds hold also in the stronger LOCAL model \cite{Linial92}, where the size of messages is not bounded.

In the distributed setting, the input to TAP is a rooted spanning tree $T$ of $G$ with root $r$, whose height is denoted by $h$. The tree $T$ is given to the vertices locally, that is, each vertex knows which of its adjacent edges is in $T$ and which of those leads to its parent in $T$.\footnote{If a root and orientation are not given, we can find a root $r$ and orient all the edges towards $r$ in $O(h)$ rounds using standard techniques.} For each vertex $v\neq r$, we denote by $p(v)$ the parent of $v$ in $T$.
The output is a set of edges $Aug$, such that $T \cup Aug$ is 2-edge-connected. In the distributed setting it is enough that at the end of the algorithm each vertex knows which of the edges incident to it are added to $Aug$.

All the messages sent in our algorithms consist of a constant number of ids, labels and weights, hence the maximal message size is bounded by $O(\log{n})$ bits, as required in the CONGEST model. \\ 

\section{A 2-approximation for Unweighted TAP in $O(h)$ rounds} \label{sec:app_uTAP}

As an introduction, we describe an $O(h)$-round 2-approximation algorithm, {\unTAP}, for unweighted TAP. 
The general structure of {\unTAP} is as follows. 

\begin{enumerate}
\itemsep0em
\item It builds a related virtual graph $G'$.
\item It finds an optimal augmentation $A'$ in $G'$.
\item It converts it to a 2-approximation augmentation $A$ in $G$.
\end{enumerate}

The graph $G'$ is defined as in \cite{khuller1993approximation}. After building $G'$, we diverge completely from the approach of \cite{khuller1993approximation} since we cannot simulate it efficiently in the distributed setting, as explained in the introduction. Instead, {\unTAP} finds an optimal augmentation in $G'$, and converts it to a 2-approximation augmentation in $G$. 
All the communication in the algorithm is on the edges of the graph $G$, since $G'$ is a virtual graph. In order to simulate the algorithm on $G$ we use labels that represent the edges of $G'$. 

In Section \ref{sec:app_build}, we describe how we build the virtual graph $G'$. Then, we show in Section \ref{sec:app_corr} that an optimal augmentation in $G'$ gives a 2-approximation augmentation in $G$. In Section 
\ref{sec:app_find}, we describe the algorithm for finding an optimal augmentation in $G'$, and we prove its correctness in Section \ref{sec:app_correct}.

\subsection{Building $G'$ from $G$} \label{sec:app_build} 

{\unTAP} starts by building a related \emph{undirected} virtual graph $G'$. Building $G'$ requires efficient computation of lowest common ancestors (LCAs), which we next explain how to obtain in the distributed setting.

\subsubsection{Computing LCAs}

We use the \textit{labeling scheme} for LCAs of Alstrup et al. \cite{alstrup2004nearest}. This labeling scheme assigns labels of size $O(\log{n})$ bits to the vertices of a rooted tree with $n$ vertices, such that given the labels of $u$ and $v$ it is possible to infer the label of their LCA. The algorithm for computing the labels takes $O(n)$ rounds in a centralized setting, and we observe that it can be implemented in $O(h)$ rounds in the distributed setting, where $h$ is the depth of the tree, as was also observed by \cite{pritchard2011fast}. This is because the algorithm consists of a constant number of traversals of the tree, from the root to the leaves or vice versa. Thus, we have:

\begin{lemma} \label{lca}
Constructing the labeling scheme for LCAs of Alstrup et al. \cite{alstrup2004nearest} takes $O(h)$ rounds.
\end{lemma}

{\unTAP} starts by applying the labeling scheme, which takes $O(h)$ rounds. We next explain how we use it in order to build $G'$.

\subsubsection{The Graph $G'$}

We next describe the graph $G'$.
To simplify the presentation of the algorithm it is convenient to give an orientation to the edges of $G'$. However, we emphasize that $G'$ is an undirected graph, that is, we do not address the notion of directed connectivity.
The graph $G'$ is defined as follows (as in \cite{khuller1993approximation}). The graph $G'$ includes all the edges of $T$, and they are all oriented towards the root $r$ of $T$. For every non-tree edge $e=\{u,v\}$ in $G$ there are two cases (see Figure \ref{pic1a}):
\begin{enumerate}
\item If $u$ is an ancestor of $v$ in $T$, we add the edge $\{u,v\}$ to $G'$, oriented from $u$ to $v$.
\item Otherwise, denote $t=LCA(u,v)$. In this case we add to $G'$ the edges $\{t,u\}$ and $\{t,v\}$, oriented from $t$ to $u$ and to $v$, respectively.
\end{enumerate}
\setlength{\intextsep}{0pt}
\begin{figure}[h]
\centering
\setlength{\abovecaptionskip}{-2pt}
\setlength{\belowcaptionskip}{6pt}
\includegraphics[scale=0.6]{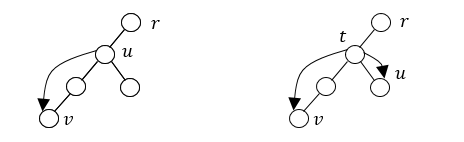}
 \caption{There are two cases for every non-tree edge in $G$. The left graph shows the first case, where the edge $\{u,v\}$ is between an ancestor and a descendant in $T$. The right graph shows the second case, where $t=LCA(u,v)$.} \label{pic1a} 
\end{figure}


Note that in the second case, the edges $\{t,u\}$ and $\{t,v\}$ added to $G'$ are not necessarily in $G$, and therefore we cannot use them for communication. Hence, the rest of the communication in the algorithm is only over the tree edges. In order to simulate the algorithm over $G'$, it is enough that each vertex knows only the tree edges incident to it (which is its input), and the labels of the non-tree edges incoming to it in $G'$. 

In order to achieve this, each vertex $v$ sends its label to all of its neighbors in $G$, and receives their labels. From them, each vertex $v$ computes the edges incoming to it in $G'$ using the labeling scheme: for each edge $e=\{u,v\}$ that is not a tree edge, $v$ uses the labels of $v$ and $u$ in order to compute $t=LCA(u,v)$. If $t=u$, i.e., $u$ is an ancestor of $v$ in $T$, the edge $\{u,v\}$ is incoming to $v$ in $G'$. Otherwise $t \neq u$, and if $t \neq v$, the edge $\{t,v\}$ is incoming to $v$ in $G'$. Since $v$ knows the labels of $u$ and $t$, using LCA computations it learns the labels of all the edges incoming to it in $G'$.

The construction of $G'$ takes $O(h)$ time, for constructing the labeling scheme by Lemma \ref{lca}. The rest of the computations take one round. This gives the following.

\begin{lemma} \label{timeb1}
Building $G'$ from $G$ takes $O(h)$ rounds.
\end{lemma}

\subsection{The Correspondence between $G$ and $G'$} \label{sec:app_corr}
We next show that an optimal augmentation in $G'$ corresponds to an augmentation in $G$ with size at most twice the size of an optimal augmentation.

To build $G'$ from $G$, for each edge $e \in G$ that is not a tree edge, we added one or two edges to $G'$. These edges are the edges \emph{corresponding} to $e$ in $G'$. Equivalently, for each such edge $\widetilde{e} \in G'$, the edge $e$ is an edge \emph{corresponding} to $\widetilde{e}$ in $G$. An edge $\widetilde{e} \in G'$ may have several corresponding edges in $G$.
A non-tree edge $e=\{u,v\}$ in $G$ covers all the edges in the unique path in $T$ between $u$ in $v$. 
We next show that the corresponding edges to $e$ in $G'$ cover together exactly the same tree edges as $e$. This allows us to show that an optimal augmentation in $G'$ gives a 2-approximation augmentation in $G$, when we replace each edge of the augmentation in $G'$ by a corresponding edge in $G$.

\begin{claim} \label{claim1}
If the non-tree edge $e=\{u,v\}$ covers the tree edge $e'$ in $G$, then one of the edges corresponding to $e$ in $G'$ covers $e'$ in $G'$.
\end{claim} 

\begin{proof}
If $e$ is in $G'$ the claim is immediate. Otherwise, the edges $\{t,u\}$ and $\{t,v\}$, where $t=LCA(u,v)$, are the edges corresponding to $e$ in $G'$. The path from $u$ to $v$ in $T$ is the union of a simple path between $u$ and $t$ and another simple path from $t$ to $v$, so the edge $e'$ must be on one of these paths, hence one of the edges $\{t,u\}$ or $\{t,v\}$ covers it.
\end{proof}

\begin{claim} \label{claim2}
If the non-tree edge $\widetilde{e}$ in $G'$ covers the tree edge $e'$, and $e$ is an edge corresponding to $\widetilde{e}$ in $G$, then $e$ covers $e'$ in $G$.
\end{claim}

\begin{proof}
If $e=\widetilde{e}$ then the claim is immediate. Otherwise, $\widetilde{e}=\{t,u\}$ for some $t,u$, and $e=\{u,v\}$ where $t=LCA(u,v)$. The edge $\widetilde{e}$ covers $e'$ in $G'$, so $e'$ is on the unique path in $T$ between $t$ and $u$. The unique path in $T$ between $u$ and $v$ is the union of a simple path between $u$ and $t$ and another simple path from $t$ to $v$. In particular, the edge $e=\{u,v\}$ covers the edge $e'$ in $G$, as needed.
\end{proof}

Assume that $A'$ is an augmentation in $G'$, and $A$ is the set of corresponding edges in $G$, where each edge in $A'$ is replaced by a corresponding edge in $G$.

\begin{corollary}
$A$ is an augmentation in $G$.
\end{corollary}

\begin{proof}
$A'$ is an augmentation so it covers all tree edges and hence from Claim \ref{claim2}, $A$ covers all tree edges, i.e., $A$ is an augmentation in $G$. 
\end{proof}

\begin{lemma} \label{corr}
Assume that $A'$ is an $\alpha$-approximation to the optimal augmentation in $G'$, then $A$ is a $2\alpha$-approximation to the optimal augmentation in $G$. 
\end{lemma}

\begin{proof}
Note that $|A|\leq |A'|$ because each edge in $A'$ is replaced by one edge in $A$. Assume that $OPT$ is an optimal augmentation in $G$ and $OPT'$ is the set of corresponding edges in $G'$, where each edge in $G$ is replaced by the corresponding one or two edges in $G'$. $OPT$ covers all tree edges, so $OPT'$ covers all tree edges by Claim \ref{claim1}, i.e, it is an augmentation in $G'$. It holds that $|OPT'|\leq 2|OPT|$ because each edge is replaced by at most two edges. Moreover, $|A'|\leq \alpha|OPT'|$ because $A'$ is an $\alpha$-approximation to the optimal augmentation in $G'$. We conclude that $$|A|\leq |A'|\leq \alpha|OPT'| \leq 2\alpha|OPT|.$$
\end{proof}

\subsection{Finding an Optimal Augmentation in $G'$} \label{sec:app_find}

The goal of {\unTAP} now is to find an optimal augmentation in $G'$.
In $G'$ all the edges that are not tree edges are between an ancestor and a descendant of it in $T$. This allows us to compare edges and define the notion of \emph{maximal} edges. Intuitively, the notion of maximal edges would capture our goal that during the algorithm, when we cover a tree edge, we would like to cover it by an edge that reaches the highest ancestor possible, allowing us to cover many tree edges simultaneously. This motivates the following definition.
 Let $v$ be a vertex in the tree, and let $e=\{u,w\}$ and $e'=\{u',w'\}$ be two edges between ancestors $u,u'$ of $v$ and descendants $w,w'$ of $v$. We say that $e$ is the \textit{maximal} edge among $e$ and $e'$ if and only if $u$ is an ancestor of $u'$. If $u=u'$ we can choose arbitrarily one of them to be the maximal edge. Among the edges incoming to $v$, the \textit{maximal} edge is the edge $\{u,v\}$ from the ancestor $u$ of $v$ that is closest to the root. Note that using the LCA labels of such edges $e,e'$, a vertex $v$ can learn which is the maximal by computing $LCA(u,u')$. Moreover, using the labels of the edge $e$, a vertex $v$ can check if $e$ covers the tree edge $\{v,p(v)\}$ using LCA computations: it checks if $v$ is an ancestor of $w$ and if $u$ is an ancestor of $p(v)$. In our algorithm, each time a vertex sends an edge $e$, it sends the labels of $e$ which allow these computations. 
 
In order to cover all tree edges of $G'$, we assign each vertex $v\neq r$ in $G'$ with the responsibility of covering the tree edge $\{v,p(v)\}$. The idea behind the algorithm is to scan the tree $T$ from the leaves to the root,  and whenever a tree edge that is still not covered is reached, it is covered by the vertex responsible for it, using the maximal edge possible.

The algorithm {\unAug} for finding an optimal augmentation in $G'$ starts at the leaves of $T$ and works as follows:
\begin{itemize}
\item Each leaf $v$ covers the tree edge $\{v,p(v)\}$ by the maximal edge $e$ incoming to $v$, it adds $e$ to the augmentation and sends $e$ to its parent. We call this a \emph{necessary} edge.
\item Each internal vertex $v$ receives from each of its children at most 2 edges: one is necessary and one is \emph{optional}. Denote by $nec_v$ the maximal necessary edge received from $v$'s children, and denote by $opt_v$ the maximal edge among all the optional edges $v$ receives from its children and the edges incoming to $v$.
There are two cases:
\begin{enumerate}
\item The tree edge $\{v,p(v)\}$ is already covered by $nec_v$. In this case $nec_v$ is the necessary edge $v$ sends to its parent. In addition, $v$ sends to its parent $opt_v$ as an optional edge. 
\item The tree edge $\{v,p(v)\}$ is not covered by $nec_v$. In this case $v$ adds to the augmentation the edge $opt_v$. From the definition of $opt_v$, it follows that it is the maximal edge that covers $\{v,p(v)\}$. In this case $opt_v$ is the edge $v$ sends to its parent as a necessary edge, and it does not send an optional edge. If $opt_v$ is an optional edge received from one of $v$'s children, $v$ updates the relevant child that this edge is necessary and has been added to the augmentation. It also updates its other children that their edges are not necessary.
\end{enumerate} 
\item When an internal vertex receives from its parent indication if the optional edge it sent is necessary, it forwards the answer to the relevant child, if such exists.
\item At the end, each vertex knows if the maximal edge incoming to it is necessary or not.
The augmentation consists of all the necessary edges.
\end{itemize}

\subsection{Correctness Proof} \label{sec:app_correct}

Denote by $A'$ the solution obtained by {\unAug}, and by $A^*$ an optimal augmentation in $G'$.

\begin{lemma} \label{opt}
The algorithm {\unAug} finds an optimal augmentation in $G'$.
\end{lemma}

\begin{proof}
First, $A'$ is an augmentation in $G'$. Consider a tree edge $e=\{v,p(v)\}$. There are edges in $G$ that cover $e$ because $G$ is 2-edge-connected, hence from Claim \ref{claim1} there are edges in $G'$ that cover $e$. Therefore, $v$ adds such an edge in order to cover $e$, if it is not already covered by $nec_v$. 

Now we show that $|A'|\leq|A^*|$, by showing a one-to-one mapping from $A'$ to $A^*$. Since $A'$ is an augmentation in $G'$, it follows that $A'$ is an optimal augmentation.

When an edge $e \in A'$ is added to $A'$ in {\unAug}, it is in order to cover some tree edge that is still
not covered, denote this edge by $t(e)$. Let $t(A')$ be all such tree edges. We map $e \in A'$ to an edge $e^* \in A^*$ that covers $t(e)$.

This mapping is one-to-one: assume to the contrary that there are two edges $e_1,e_2 \in A'$ that are mapped to the same edge $e^* \in A^*$. Note that $e^*$ is an edge between an ancestor and its descendant in $T$ that covers both $t(e_1)=\{v_1,p(v_1)\}$ and $t(e_2)=\{v_2,p(v_2)\}$. Hence, $t(e_1)$ and $t(e_2)$ are on the same path in the tree between an ancestor and its descendant. Assume that $t(e_2)$ is closer to the root $r$ on this path. Note that the tree edge $t(e_1)$ is not covered by $nec_{v_1}$ since $t(e_1) \in t(A')$. Hence, $v_1$ adds the edge $e_1$ in order to cover it, which is the maximal edge possible. Since the edge $e^*$ covers both $t(e_1)$ and $t(e_2)$, it follows that $e_1$ covers $t(e_2)$ as well, contradicting the fact that $t(e_2) \in t(A')$. This completes the proof that  $|A'|\leq|A^*|$.
\end{proof}

We complete {\unTAP} by replacing each edge in $A'$ by a corresponding edge in $G$.

\begin{lemma} \label{time}
The time complexity of {\unTAP} is $O(h)$ rounds.
\end{lemma}
\begin{proof}
Building $G'$ from $G$ takes $O(h)$ rounds by Lemma \ref{timeb1}. Finding an optimal augmentation in $G'$ takes $O(h)$ rounds as well: the algorithm {\unAug} consists of two traversals of the tree, from the leaves to the root, and vice versa. Hence, the total time complexity of {\unTAP} is $O(h)$ rounds.
\end{proof}

\begin{restatable}{theorem}{uTAP} \label{uTAP}
There is a distributed 2-approximation algorithm for unweighted TAP in the CONGEST model that runs in $O(h)$ rounds, where $h$ is the height of the tree $T$. 
\end{restatable}


\begin{proof}
The algorithm {\unAug} finds an optimal augmentation in $G'$, as proven in Lemma \ref{opt}. By Lemma \ref{corr}, this corresponds to an augmentation in $G$ with size at most twice the optimal augmentation of $G$. The time complexity follows from Lemma \ref{time}.
\end{proof}

\section{A 2-approximation for Weighted TAP in $O(h)$ rounds} \label{sec:app_wTAP}

In this section, we prove Theorem \ref{wTAP}.

\wTAP*

Our algorithm for weighted TAP, {\weTAP}, has the same structure of {\unTAP}. It starts by building the same virtual graph $G'$, and then it finds an optimal augmentation in $G'$. The only difference in building $G'$ is that now each edge $e$ is replaced by one or two edges in $G'$ with the same weight that $e$ has. The proof that an optimal augmentation in $G'$ corresponds to an augmentation in $G$ with at most twice the cost of an optimal augmentation in $G$ is the same as in the unweighted case.

The difference is in finding an optimal augmentation in $G'$. In the unweighted case, for each vertex $v$, the only edge incoming to $v$ in $G'$ that was useful for the algorithm was the maximal edge. However, when edges have weights, potentially all the edges incoming to $v$ may be useful for the algorithm, and we can no longer use the notion of \emph{maximal} edges in order to compare edges. 
This is because of the tension between heavy edges that cover many edges of the tree, and light edges that cover less edges of the tree.
To overcome this obstacle, we introduce a new technique of \emph{altering} the weights of the edges we send in the algorithm. 

Let $min_v$ be the weight of the minimum weight edge that covers $\{v,p(v)\}$.
The intuition behind our approach is that in order to cover the tree edge $\{v,p(v)\}$ we must pay at least $min_v$. Thus, $min_v$ captures the cost of covering this tree edge. 
Therefore, before sending to its parent information about relevant edges, $v$ alters their weights by reducing from them the weight $min_v$. 
We show that altering the weights is crucial for selecting which edges to add to the augmentation, and allows to divide the weight of an edge in a way that captures the cost for covering each tree edge. In addition, we show that using this approach, sending information about at most $h$ edges from each vertex to its parent suffices for selecting the best edges for the augmentation.

In Section \ref{sec:app_alg}, we describe our algorithm for finding an optimal augmentation in $G'$. In Section \ref{sec:app_wcorrect}, we prove the correctness of the algorithm, and in Section \ref{sec:app_analysis}, we analyze its time complexity. 

\subsection{Finding an Optimal Augmentation in $G'$} \label{sec:app_alg}

Our algorithm consists of two traversals of the tree: from the leaves to the root and vice versa. As in {\unAug}, each vertex $v$ is responsible for covering the tree edge $\{v,p(v)\}$.

In the first traversal, each vertex $v$ computes the weight $min_v$ of the minimum weight edge that covers the tree edge $\{v,p(v)\}$ according to the weights of the edges it receives from its children, and the weights of the edges incoming to it. It also computes the weights of the minimum weight edges that cover the path from $v$ to each of its ancestors $u$, according to the weights $v$ receives in the algorithm. Then, $v$ subtracts $min_v$ from the weights of these edges, and sends them to its parent with the altered weights.

In the second traversal, we scan the tree from the root to the leaves. Each child $v$ of $r$ adds to the augmentation the edge having weight $min_v$. It informs the relevant child who sent it, if exists, and informs its other children it did not add their edges. Each internal vertex $v$ receives from its parent a message that indicates whether one of the edges it sent was added to the augmentation by one of its ancestors or not. In the former case, $v$ learns that this edge was added to the augmentation and forwards the message to the relevant child who sent it, if such exists. Otherwise, the tree edge $\{v,p(v)\}$ is still not covered, and $v$ adds to the augmentation the edge having weight $min_v$. It informs the relevant child who sent it, if exists, and informs its other children that their edges were not added to the augmentation.

A description of the algorithm is given in Algorithm \ref{alg}. For simplicity of presentation, we start by describing an algorithm which takes $O(h^2)$ rounds. Later, in Section \ref{sec:app_analysis}, we explain how using pipelining we improve the time complexity to $O(h)$ rounds. \\

\begin{algorithm}
\caption{Finding an Optimal Augmentation in $G'$}\label{alg}
\begin{algorithmic}[1]
\Statex
\Statex The code is for every vertex $v \neq r$
\Statex
\State \underline{Initialization:} 
\State $e_{v,u} \gets$ the minimum weight edge incoming to $v$ that covers the path between $v$ and its ancestor $u$ or $\bot$ if there is no such edge. 
\State $w_v(u) \gets w(e_{v,u})$ for each ancestor $u$ of $v$ such that $e_{v,u} \neq \bot$, and $w_v(u) \gets \infty$ otherwise. 
\State $A_v \gets$ the union of $v$ and its children in $T$.
\State $Aug_v \gets \emptyset$
\Statex
\State \underline{First Traversal:} 
\If {$v$ is a leaf}
\State for each ancestor $u$ of $v$:
$sender_v(u) \gets v$
\Else
\State \textbf{wait} for receiving $w_{v'}(u)$ for all ancestors $u$ of $v$, from each child $v'$ of $v$
\State for each ancestor $u$ of $v$: $w_v(u) \gets min_{v' \in A_v} {w_{v'}(u)}$, 
$sender_v(u) \gets argmin_{v' \in A_v} {w_{v'}(u)}$\label{l1} 
\EndIf 
\State $min_v \gets w_v(p(v))$ \label{line}
\State for each ancestor $u$ of $v$:
$w_v(u) \gets w_v(u)-min_v$
\State for each ancestor $u \neq p(v)$ of $v$ \textbf{send} $(u,w_v(u))$ to $p(v)$
\Statex
\State \underline{Second Traversal:} 
\State $u \gets p(v)$
\If {$v$ is not a child of $r$}
\State \textbf{wait} for a message $m$ from $p(v)$
\If {$m \neq \bot$} $u \gets m$
\EndIf
\EndIf
\State $s \gets sender_v(u)$
\If {$s=v$} 
\State $Aug_v \gets Aug_v \cup \{e_{v,u}\}$ \label{take}
\Else 
\State send $u$ to $s$
\EndIf
\State for each child $v' \neq s$ of $v$ \textbf{send} $\bot$ to $v'$
\end{algorithmic}
\end{algorithm} 

\subsubsection*{Technical Details:}

We assume in the algorithm that each vertex knows all the ids of its ancestors in $T$. We justify it in the next claim.
Note that when we construct $G'$, if $\{u,v\}$ is an edge between an ancestor $u$ and its descendant $v$ in $T$, $v$ learns the label of $u$ according to the LCA labeling scheme and not the id of $u$. However, once $v$ learns about the ids and labels of all its ancestors, it knows the id of $u$ as well, and can use it in the algorithm.

\begin{claim}
All the vertices can learn the ids and labels of all their ancestors in $O(h)$ rounds.
\end{claim}

\begin{proof}
In order to do this, at the first round each vertex sends to its children its id and label. In the next round, each vertex sends to its children the id and label it received in the previous round, and we continue in the same way until each vertex learns about all its ancestors. Clearly, after $h$ rounds each vertex learns all the ids and labels of all its ancestors.
\end{proof}

\begin{claim}
If a vertex $v$ adds $e_{v,u}$ to $Aug_v$ in line \ref{take} of its algorithm, then $e_{v,u} \neq \bot$.
\end{claim}

\begin{proof}
Since $G'$ is 2-edge-connected, we can cover all tree edges by edges from $G'$. Hence, the minimum weight of an edge that covers some tree edge is never infinite. It follows that if a vertex $v$ adds $e_{v,u}$ to $Aug_v$, then $e_{v,u} \neq \bot$.
\end{proof}

\subsection{Correctness Proof} \label{sec:app_wcorrect}

The challenge in establishing the correctness of our algorithm lies in the fact that the vertices use altered weights rather than the original ones. Nevertheless, we show that our intuition behind choosing these altered weights faithfully captures the essence of finding an augmentation in the weighted case.

\begin{lemma} \label{wcor}
Algorithm \ref{alg} finds an optimal augmentation in $G'$.
\end{lemma}

\begin{proof}
Note that the solution obtained by the algorithm is an augmentation of $G'$ because each vertex $v$ adds an edge in order to cover the tree edge $\{v,p(v)\}$ if it is not already covered by an edge which one of its ancestors decides to add to the augmentation.

We next show that the augmentation is optimal. The key ingredient we use in our proof is giving costs to the edges of $T$ such that the sum of the costs is equal to both the cost of the solution obtained by the algorithm and the cost of an optimal augmentation of $G'$. 
Hence, we conclude that the cost of the solution obtained by the algorithm is optimal.

\paragraph*{Giving costs to the edges of $T$:$\ $} Fix a vertex $v \neq r$ and let $t=\{v,p(v)\}$. We define $c(t)=min_v$ (the value of $w_v(p(v))$ in line \ref{line} of the algorithm).

For an edge $e=\{u,x\}$ that covers $t$, such that $u$ is an ancestor of $x$ in $T$, let $P$ be the path of tree edges between $x$ and $p(v)$ in $T$. Note that the path $P$ is defined with respect to $t$ and $e$. For a vertex $v'$ such that $\{v',p(v')\} \in P$, let $P_{v'}$ be the path of tree edges between $x$ and $v'$. 
Note that $min_v$ is the weight of the minimum weight edge covering the tree edge $t=\{v,p(v)\}$ according to the weights $v$ receives in the algorithm. Denote this edge by $e_v$.

\begin{claim} \label{c1}
$w(e_v)=\sum_{t' \in P} c(t')$, where $P$ is the path defined by $t=\{v,p(v)\}$ and $e_v$.
\end{claim} 

\begin{proof}
Let $e_v=\{u,x\}$, where $u$ is an ancestor of $x$ in $T$. For each vertex $v'$ on the path between $x$ and $v$, $e_v$ is the minimum weight edge covering the path between $v'$ and its ancestor $p(v)$, according to the weights $v'$ receives in the algorithm, as otherwise we get a contradiction to the definition of $e_v$. Each vertex on this path reduces $min_{v'}$ from the weight of $e_v$ it receives before sending it to its parent. Denote by $V'$ all the vertices on the path between $x$ and $v$, excluding $v$. It follows that $$c(t)=min_v=w(e_v)-\sum_{v' \in V'} min_{v'}=w(e_v)-\sum_{t' \in P_v}c(t'),$$ which gives $w(e_v)=\sum_{t' \in P} c(t')$.
\end{proof}

\begin{claim} \label{c2}
For each edge $e$ that covers $t$, it holds that $w(e) \geq \sum_{t' \in P} c(t'),$ where $P$ is the path defined by $t$ and $e$.
\end{claim}

\begin{proof}
Let $e=\{u,x\}$ be an edge that covers $t=\{v,p(v)\}$ where $u$ is an ancestor of $x$ in $T$. Denote by $P_v=\{x=v_1,...,v_k=v\}$ the path of tree edges between $x$ and $v$ in $T$. 
We prove by induction that $$w_{v_i}(p(v)) \leq w(e)-\sum_{t' \in P_{v_i}} c(t'),$$ where $w_{v_i}(p(v))$ is the value obtained in line \ref{l1} of the algorithm of $v_i$ (or at the initialization if $v_i$ is a leaf).  

For $i=1$, let $e_{v_1,p(v)}$ be the minimum weight edge incoming to $v_1$ that covers the path between $v_1$ and $p(v)$ in $T$. Note that $w(e_{v_1,p(v)}) \leq w(e)$ because $e$ is an edge incoming to $v_1$ that covers the path between $v_1$ and $p(v)$. The value of $w_{v_1}(p(v))$ is the weight of the minimum weight edge covering the path between $v_1$ and $p(v)$, according to the weights $v_1$ receives. In particular, $w_{v_1}(p(v)) \leq w(e_{v_1,p(v)})$, and therefore $w_{v_1}(p(v)) \leq w(e)$. Since $P_{v_1}$ is an empty path, we have $\sum_{t' \in P_{v_1}} c(t') = 0$, which gives $$w_{v_1}(p(v)) \leq w(e)-\sum_{t' \in P_{v_1}} c(t').$$

Assume the claim holds for $i$, and we prove it holds for $i+1$. Denote by $t_i$ the tree edge $\{v_i,v_{i+1}\}$. Note that $v_i$ sends to $v_{i+1}$ the message $(p(v), w_{v_i}(p(v))-min_{v_i})$ since it reduces $min_{v_i}$ from the value of $w_{v_i}(p(v))$ before sending it to its parent. The value of $w_{v_{i+1}}(p(v))$ is the weight of the minimum weight edge covering the path between $v_{i+1}$ and $p(v)$, according to the weights $v_{i+1}$ receives. In particular, $w_{v_{i+1}}(p(v)) \leq w_{v_i}(p(v))-min_{v_i}$.  
By the induction hypothesis $w_{v_i}(p(v)) \leq w(e)-\sum_{t' \in P_{v_i}} c(t')$, which gives $$w_{v_{i+1}}(p(v)) \leq w(e)-\sum_{t' \in P_{v_i}} c(t') - min_{v_i} = w(e) - \sum_{t' \in P_{v_{i+1}}} c(t').$$ 
For $i=k$ we get $$c(t) = w_{v}(p(v)) \leq w(e)-\sum_{t' \in P_v} c(t'),$$
which implies that $w(e) \geq \sum_{t' \in P} c(t')$, as claimed.
\end{proof}

\begin{claim} \label{c3}
The sum of the costs of the edges of $T$ is equal to the cost of the solution obtained by the algorithm.
\end{claim}

\begin{proof}
We map each edge $e$ added to the augmentation to a path $P_e$ of tree edges, such that:
\begin{enumerate}[(I)]
\item The paths that correspond to different augmentation edges are disjoint, and their union is the entire tree $T$. That is, $P_e \cap P_{e'} = \emptyset$ for $e \neq e'$, and $\cup P_e = T$.
\item The weight of $e$ is equal to the sum of costs of tree edges in the corresponding path, i.e., $w(e)=\sum_{t' \in P_e} c(t')$.
\end{enumerate}
Let $e=\{u,x\}$ be an edge added to the augmentation, such that $u$ is an ancestor of $x$ in $T$. Let $v$ be the vertex that decides to add $e$ to the augmentation. Note that $v$ decides to add $e$ to the augmentation because it covers the tree edge $\{v,p(v)\}$, which is not covered yet by an edge that one of $v$'s ancestors decides to add to the augmentation. We map $e$ to the tree path $P_e$ that consists of all the tree edges on the path between $x$ and $p(v)$. Note that $e$ covers all the edges on this path (and it may also cover other tree edges, on the path between $p(v)$ and $u$ in $T$). This divides the tree edges to disjoint paths because the vertices on the path between $x$ and $p(v)$ do not decide to add other edges to the augmentation, since all the relevant tree edges are already covered by $e$. In addition, these paths include all tree edges because the edges added to the augmentation cover all tree edges. This proves (I).

Note that $v$ adds $e$ to the augmentation because the tree edge $\{v,p(v)\}$ is not covered yet. So $v$ chooses $e$ because it is the minimum weight edge $e_v$ that covers $\{v,p(v)\}$. By Claim \ref{c1}, it holds that $w(e_v)=\sum_{t' \in P} c(t')$ where $P=P_e$ is the path of tree edges between $x$ and $p(v)$. This proves (II).
(I) and (II) complete the proof that the cost of all the edges added to the augmentation is equal to the sum of costs of the edges in $T$.
\end{proof}

\begin{claim} \label{c4}
The cost of any augmentation of $G'$ is at least the sum of costs of the edges of $T$.
\end{claim}

\begin{proof}
Let $A$ be an augmentation in $G'$. We map a subset of edges $E' \subseteq A$ to paths $\{P'_e\}_{e \in E'}$ in $T$ such that:
\begin{enumerate}[(I)]
\item The paths that correspond to different edges are disjoint, and their union is the entire tree $T$.
\item The weight of an edge $e \in E'$ is at least the sum of costs of tree edges on the path $P'_e$.
\end{enumerate}
We cover tree edges by edges from $A$ as follows. While there is a tree edge that is still not covered, we choose a tree edge $\{v,p(v)\}$ that is still not covered and is closest to the root $r$, where initially $p(v)=r$. Since $A$ is an augmentation, there is an edge $e=\{u,x\}$ in $A$ such that $u$ is an ancestor of $x$ in $T$ and $e$ covers $\{v,p(v)\}$. We map $e$ to the tree path $P'_e$ between $x$ and $p(v)$. The edge $e$ covers all the tree edges on this path, and may cover additional edges closer to the root that are already covered by other edges from $A$. We continue in the same manner until all the tree edges are covered. From the construction, the paths are disjoint and include all tree edges, proving (I).

From Claim \ref{c2}, it holds that $w(e) \geq \sum_{t' \in P} c(t')$ where $P=P'_e$ is the path of tree edges between $x$ and $p(v)$, proving (II).

To conclude, the cost of all the edges in $A$ is at least the sum of costs of all the edges of $T$. Note that there might be edges from $A$ that are not mapped to paths in $T$, which can only increase the cost of $A$.
\end{proof}

From Claims \ref{c3} and \ref{c4} we have that the cost of the augmentation obtained by the algorithm is smaller or equal to the cost of any augmentation of $G'$, hence the solution obtained by the algorithm is optimal. This completes the proof of Lemma \ref{wcor}.
\end{proof}

\subsection{Time analysis} \label{sec:app_analysis}

We next analyze the time complexity of the algorithm. In the second traversal of the tree, each parent sends to each of its children one message, which takes $O(h)$ rounds. In the first traversal of the tree, each vertex sends to its parent at most $h$ edges. 
If each vertex waits to receive all the messages from its children, before sending messages to its parent, it would result in a time complexity of $O(h^2)$ rounds. However, using pipelining we get a time complexity of $O(h)$ rounds. To show this, we carefully design each vertex $v$ to send the messages $(u,w_v(u))$ in increasing order of heights of its ancestors.

The main intuition is that although each vertex $v$ may receive $h$ different messages from each of its children during the algorithm, in order for $v$ to send to its parent $p(v)$ the message concerning an ancestor $u$, the vertex $v$ only needs to receive one message from each of its children concerning the ancestor $u$. Hence, if all the vertices send the messages according to increasing order of heights of their ancestors, we can pipeline the messages and get a time complexity of $O(h)$ rounds. We formalize this intuition in the next lemma.

\begin{lemma}
If all the vertices send the messages according to increasing order of heights of their ancestors, the following holds.
A vertex $v$ at height $i$ sends to its parent until round $i+j$ the message $(u,w_v(u))$ such that $u$ is an ancestor of $v$ at height $j$.
\end{lemma}

\begin{proof}
We prove the lemma by induction. For a vertex at height 0 (a leaf) the claim holds since $v$ sends the messages according to increasing order of heights. We assume that the claim holds for each vertex at height at most $i-1$, and show that it also holds for each vertex $v$ at height $i$. 

If $j \leq i$ the claim holds trivially, since $v$ does not have ancestors at height $j$. 
We assume that the claim holds for $i$ and $j-1$ and we show that it also holds for $i$ and $j$. Let $v$ be a vertex at height $i$, and let $u$ be an ancestor of $v$ at height $j$. Note that by the induction hypothesis, by round $i-1+j$ all the children $v'$ of $v$ already sent to $v$ the messages $(u,w_{v'}(u))$.
Therefore, $v$ can compute $w_v(u) \gets min_{v' \in A_v} {w_{v'}(u)}$.
Note that by round $i+j-1$, $v$ already sent all the messages concerning ancestors at height at most $j-1$ and sends the message concerning $u$ to its parent until round $i+j$ as needed (in the case that $u=p(v)$ no message is sent in the algorithm). Note that $v$ also knows and sends the new weight $w_v(u)$: denote by $i'$ the height of the parent of $v$ ($i < i'$), then each other ancestor of $v$ is at height greater than $i'$. Until round $i+i'$, $v$ knows $min_v=w_v(p(v))$, so for all the relevant values of $j$ ($i' \leq j$) it can compute the new weight $w_v(u) \gets w_v(u)-min_v$ until round $i+j$. 
\end{proof}

From the lemma we get that by round $2h$ all the children of $r$ learn about the minimum weight edge that covers the tree edge between them and $r$, so the first traversal is completed after $O(h)$ rounds. It follows that the overall time complexity of the algorithm is $O(h)$ rounds as needed, giving the following.

\begin{lemma} \label{timew}
Algorithm \ref{alg} completes in $O(h)$ rounds.
\end{lemma}

\wTAP*

\begin{proof}
By Lemma \ref{wcor}, Algorithm \ref{alg} finds an optimal augmentation in $G'$. Its time complexity is $O(h)$ rounds by Lemma \ref{timew}. This augmentation corresponds to an augmentation in $G$ with cost at most twice the cost of an optimal augmentation of $G$ by Lemma \ref{corr} (the proof is for the unweighted case, but the same proof shows it holds for the weighted case as well). Building $G'$ is the same as in the unweighted case and takes $O(h)$ rounds by Lemma \ref{timeb1}.
\end{proof}

\section{Applications} \label{sec:applic}

In this section, we discuss applications of our algorithms, and show they provide efficient algorithms for additional related problems. \\

\textbf{Minimum Weight 2-Edge-Connected Spanning Subgraph:}
In the minimum weight 2-edge-connected spanning subgraph problem (2-ECSS), the input is a 2-edge-connected graph $G$, and the goal is to find the minimum weight 2-edge-connected spanning subgraph of $G$. Using {\unTAP} we have the following.

\ECSS*

\begin{proof}
We apply {\unTAP} on $G$ and a BFS tree $T$ of $G$. Finding a BFS tree takes $O(D)$ rounds \cite{peleg2000distributed}, and {\unTAP} takes $O(D)$ rounds since $T$ is a BFS tree. The size of the augmentation $Aug$ is at most $n-1$ because in the worst case we add a different edge in order to cover each tree edge. Hence, $T \cup Aug$ is a 2-edge-connected subgraph with at most $2(n-1)$ edges. Note that any 2-edge-connected graph has at least $n$ edges, which implies a 2-approximation, as claimed.
\end{proof}

The above algorithm has a better time complexity compared to the algorithm of \cite{krumke2007distributed}, which finds a $\frac{3}{2}$-approximation to 2-ECSS in $O(n)$ rounds. In the algorithm of \cite{krumke2007distributed}, the augmented tree $T$ is a DFS tree rather then a BFS tree. The same proof of \cite{krumke2007distributed, khuller1994biconnectivity} gives that if we apply {\unTAP} on $G$ and a DFS tree we also obtain a $\frac{3}{2}$-approximation to 2-ECSS in $O(n)$ rounds.
For weighted 2-ECSS, using {\weTAP} gives the following.

\begin{theorem}
There is a distributed 3-approximation algorithm for weighted 2-ECSS in the CONGEST model that completes in $O(h_{MST}+\sqrt{n}\log^*{n})$ rounds, where $h_{MST}$ is the height of the MST.
\end{theorem}

\begin{proof}
We follow the same approach of \cite{krumke2007distributed}. We start by constructing an MST, which takes $O(D+\sqrt{n}\log^*{n})$ rounds \cite{kutten1995fast}, and then we augment it using {\weTAP} in $O(h_{MST})$ rounds.\footnote{We assume that the MST is unique. Otherwise, $h_{MST}$ is the height of the MST we construct.} Let $w(A)$ be the weight of an optimal solution $A$ to weighted 2-ECSS. Since both the MST and an optimal augmentation have weights at most $w(A)$, and since our algorithm for weighted TAP gives a 2-approximation, this approach gives a 3-approximation for weighted 2-ECSS. 
\end{proof}

This algorithm has a better time complexity compared to the algorithm of \cite{krumke2007distributed}, which takes $O(n \log{n})$ rounds, with the same approximation ratio. \\

\textbf{Increasing the Edge-Connectivity from 1 to 2:}
{\weTAP} is a 2-approximation algorithm for TAP, but can also be used to increase the connectivity of any spanning subgraph $H$ of $G$ from $1$ to $2$. In order to do so, we start by finding a spanning tree $T$ of $H$. Note that it is not enough to apply {\unTAP} on $T$ and take the augmentation obtained, since edges from $H$ can be added to the augmentation with no cost. Hence, we apply {\weTAP} on $G$ and $T$, where we set the weights of all the edges of $H$ to be $0$.
The augmentation $Aug$ we obtain is a set of edges such that $T \cup Aug$ is 2-edge-connected, which also implies that $H \cup Aug$ is 2-edge-connected. In addition, its cost is at most twice the cost of an optimal augmentation of $H$, because any augmentation of $H$ corresponds to an augmentation of $T$ with the same cost, and $Aug$ is a $2$-approximation to the optimal augmentation of $T$. The time complexity is $O(D_H)$ rounds where $D_H$ is the diameter of $H$, since finding a spanning tree $T$ of $H$ takes $O(D_H)$ rounds and applying {\weTAP} takes $O(D_H)$ rounds because it is the height of $T$. \\

\textbf{Verifying 2-Edge-Connectivity:}
The algorithm {\unTAP} can be used in order to verify if a connected graph $G$ is 2-edge-connected in $O(D)$ rounds, where at the end of the algorithm all the vertices know if $G$ is 2-edge-connected.\footnote{A verification algorithm with the same complexity can also be deduced from the edge-biconnectivity algorithm of Pritchard \cite{pritchard2005robust}.} We start by building a BFS tree $T$ of $G$ and then apply {\unTAP} to $G$ and $T$. Note that when we find an optimal augmentation in $G'$ by {\unAug}, each vertex $v$ is responsible to cover the tree edge $\{v,p(v)\}$. If the graph $G$ is 2-edge-connected, all the edges can be covered.
If the graph $G$ is not 2-edge-connected, then there is a tree edge $\{v,p(v)\}$ that is a bridge in the graph, and hence cannot be covered by any edge in $G$. In such a case, $v$ identifies that it cannot cover the edge and hence the graph is not 2-edge-connected. 
Therefore, after scanning the tree from the leaves to the root in {\unAug}, we can distinguish between these two cases, which takes $O(D)$ rounds. The root $r$ can distribute the information to all the vertices in $O(D)$ rounds as well.

\section{A 4-approximation for Unweighted TAP in $\widetilde{O}(D+\sqrt{n})$ rounds} \label{sec:faster}

The time complexity of {\unTAP} and {\weTAP} is linear in the height of $T$. When $h$ is large, we suggest a much faster $O(D+\sqrt{n}\log^*{n})$-round algorithm for unweighted TAP, proving Theorem \ref{uTAPtwo}.

\uTAPtwo*




The structure of the algorithm is the same as the structure of {\unTAP}. It starts by building the same virtual graph $G'$, and then it finds an augmentation in $G'$. 
However, now we do not necessarily obtain an optimal augmentation in $G'$, but rather a 2-approximation to the optimal augmentation of $G'$, which results in a 4-approximation to the optimal augmentation in $G$.
Since we want to reduce the time complexity, our algorithm cannot scan the whole tree anymore. Therefore, we can no longer use directly the LCA labeling scheme and the algorithm {\unAug} for finding an optimal augmentation. To overcome this, we break the tree $T$ into fragments, and we divide the algorithm into local parts, in which we communicate in each fragment separately, and to global parts, in which we coordinate between different fragments over a BFS tree. This approach is useful also in other distributed algorithms for global problems, such as finding an MST \cite{kutten1995fast} or a minimum cut \cite{nanongkai2014almost}.
The challenge is showing that this approach guarantees a good approximation.
Since our algorithm does not scan the whole tree $T$ it may add different edges in order to cover the same tree edges, which makes the analysis much more involved. \\

\textbf{Building $G'$ from $G$:}
To build $G'$ from $G$ we use the labeling scheme for LCAs that we used in {\unTAP}. However, applying this scheme directly takes $O(h)$ rounds. We show how to compute all the relevant LCAs more efficiently in $O(D + \sqrt{n})$ rounds. The idea is to apply the labeling scheme on each fragment separately to obtain \emph{local labels}, and to apply the labeling scheme on the tree of fragments to obtain \emph{global labels}. We show that using the local and global labels, and additional information on the structure of the tree of fragments, each vertex can compute all the edges incoming to it in $G'$. \\ 

\textbf{Finding an augmentation in $G'$:}
In order to find an augmentation in $G'$, we need to cover tree edges between fragments (\emph{global edges}) and tree edges in the same fragment (\emph{local edges}). 
We next give a high-level overview of our approach, the exact algorithm differs slightly from this description and appears in Section \ref{sec:app_aug}.
We start by computing all the maximal edges that cover the global edges. To cover all the global edges, one approach could be to add all these maximal edges to the augmentation. However, this cannot guarantee a good approximation.
Instead, we apply {\unAug} on the tree of fragments in order to cover all the global edges. Then, we apply it on each fragment separately in order to cover the local edges in the fragment that are still not covered. This algorithm requires coordination between different fragments, since each vertex $v$ needs to learn if the tree edge $\{v,p(v)\}$ is already covered after the first part of the algorithm. In addition, although the second part is applied on each fragment separately, a vertex $v$ may need to add an edge incoming to another fragment to cover the tree edge $\{v,p(v)\}$. For achieving an efficient time complexity, we show how to use only $O(\sqrt{n})$ different messages for the whole coordination of the algorithm. \\

We next provide full details of the algorithm. In Section \ref{frag}, we explain how we break the tree into fragments using the MST algorithm of Kutten and Peleg \cite{kutten1995fast}. In Section \ref{sec:app_lca}, we show how we build the graph $G'$, and in Section \ref{sec:app_aug} we explain how we find an augmentation in $G'$. The approximation ratio analysis appears in Section \ref{sec:app_approx}.

\subsection{Breaking $T$ into fragments} \label{frag}
We break the tree $T$ into fragments, such that each fragment is a tree with diameter at most $O(\sqrt{n})$ and there are at most $O(\sqrt{n})$ fragments. We do this by using the MST algorithm of Kutten and Peleg \cite{kutten1995fast} which has a time complexity of $O(D+\sqrt{n}\log^*{n})$ rounds. We say that a tree edge is a \textit{local} edge if its vertices are in the same fragment, and is a \textit{global} edge if it connects two fragments. The tree of fragments $T_F$ is the tree obtained by contracting each fragment $F$ into one vertex $v_F$ and having an edge between $v_{F_1}$ and $v_{F_2}$ if the two fragments are connected by a global edge. Since there are at most $O(\sqrt{n})$ fragments, $T_F$ is of size $O(\sqrt{n})$. Each fragment has a root, which is the vertex $v$ closest to $r$ in the fragment. 

Our algorithm is divided to local parts, in which we communicate in each fragment separately, which results in time complexity proportional to the fragments' diameter, $O(\sqrt{n})$, and to global parts, in which we coordinate between different fragments over a BFS tree rooted at $r$. Building a BFS tree rooted at $r$ takes $O(D)$ rounds \cite{peleg2000distributed}. Using the BFS tree we can distribute $k$ different messages from vertices in the tree to all the vertices in the tree in $O(D+k)$ rounds: we first collect all the messages in the root $r$ using upcast, and then $r$ broadcasts the messages to all the vertices in the tree. Each of these parts takes $O(D+k)$ rounds \cite{peleg2000distributed}. We show that it is enough to distribute $O(\sqrt{n})$ different messages for the coordination, which results in time complexity of $O(D + \sqrt{n})$ rounds. The overall time complexity of the algorithm is $O(D+\sqrt{n}\log^*{n})$ rounds.

\subsection{Building $G'$ from $G$} \label{sec:app_lca}

In order to build $G'$ from $G$, it is enough that each vertex knows all the edges incoming to it in $G'$. In order to obtain this, we use the labeling scheme for LCAs that we used for {\unTAP}. However, applying this scheme takes $O(h)$ rounds, and in order to avoid the dependence on $h$ we break the label to a local part and a global part in the following way:
\begin{itemize}
 \item We first apply the labeling scheme for LCAs on each fragment separately, to obtain local labels.
 \item Next, we apply the labeling scheme for LCAs on $T_F$, such that each fragment gets a label. This is the global label of all the vertices in the fragment. 
\end{itemize} 

The first part takes $O(h_F)$ rounds on a fragment $F$ of height $h_F$. Since the diameter of each fragment is  $O(\sqrt{n})$, it follows that this part takes $O(\sqrt{n})$ rounds.

In order to implement the second part efficiently, we first distribute information about the global edges to all the vertices. Note that each global edge connects two fragments. We assume that each fragment has an id known to all the vertices in the fragment, say, the id of the root of the fragment, which it can distribute to all the vertices in the fragment in $O(\sqrt{n})$ rounds. For each global edge $e=\{v,p(v)\}$, the vertex $v$ distributes the message $(id_1,id_2,\ell_1,\ell_2)$ where $id_1,id_2$ are the ids of the fragments of $v$ and $p(v)$, and $\ell_1$, $\ell_2$ are the local labels of $v$ and $p(v)$. Since there are $O(\sqrt{n})$ global edges, we can distribute this information over the BFS tree to all the vertices in $O(D + \sqrt{n})$ rounds.
After distributing the information about the global edges to all the vertices, they all learn the whole structure of $T_F$. Now each vertex can compute locally the labeling scheme for LCAs on $T_F$ and, in particular, learn its global label. Note that applying the labeling scheme does not require communication, so the total round complexity of the second part is $O(D + \sqrt{n})$.

We now explain how we use the local and global labels in order to compute LCAs in $T$. Assume the vertices $v,u$ have the local labels $\ell_v,\ell_u$ and the global labels $g_v,g_u$, respectively:
\begin{itemize}
 \item If $g_v=g_u$ then $v$ and $u$ are in the same fragment. It follows that their LCA is in this fragment, since the root of the fragment is an ancestor of both of them. In this case we use the local labels $\ell_v,\ell_u$ in order to compute the local label of their LCA in the fragment, whose global label is $g_v$.
 \item If $g_v \neq g_u$ then $v$ and $u$ are in different fragments $F_v, F_u$. They use the global labels in order to compute the global label $g$ of the fragment $F$ that is the LCA of $F_v,F_u$ in $T_F$. In this case it follows that the LCA of $v$ and $u$ in $T$ is in $F$, and its global label is $g$. If $F=F_v$ it follows that $v$ is the LCA of $v$ and $u$, so its local label is $\ell_v$. Similarly, if $F=F_u$ then its local label is $\ell_u$. Otherwise, in order to find the local label of the LCA, note that $v$ and $u$ know the whole structure of $T_F$. In particular, they can find the paths between $F_v$ to $F$ in $T_F$, and between $F_u$ to $F$ in $T_F$. The last edges on these paths are global edges of the form $e_1=\{v_1,p(v_1)\},e_2=\{v_2,p(v_2)\}$ where $p(v_1),p(v_2)$ are in $F$ ($e_1 \neq e_2$, otherwise we get a contradiction to the fact that $F$ is the LCA of $F_v,F_u$ in $T_F$). Note that $v$ and $u$ know the local labels of all the vertices in global edges, and in particular they know the local labels $\ell_1,\ell_2$ of $p(v_1),p(v_2)$. They can use $\ell_1,\ell_2$ in order to compute the local label of the LCA of $p(v_1),p(v_2)$ in $F$. This is the LCA of $v,u$ in $T$. In conclusion, using $g_v,g_u$, $v$ and $u$ can compute the local and global labels of their LCA in $T$. The computation is based on the information about global edges all vertices know, and does not require communiction.    
\end{itemize} 

We explained how all the vertices get local and global labels, and how they use these labels in order to compute LCAs in $T$. As in {\unTAP}, in one round each vertex can send its labels to all its neighbors in $G$, and get their labels. From these labels each vertex can compute the local and global labels of all the edges incoming to it in $G'$ by computing LCAs, which does not require communication. The overall time complexity of constricting $G'$ is $O(D + \sqrt{n})$ rounds, for applying the labeling scheme. This gives,

\begin{lemma} \label{timeb}
Building $G'$ from $G$ takes $O(D + \sqrt{n})$ rounds.
\end{lemma}

\subsection{Finding an Augmentation in $G'$} \label{sec:app_aug}

We next explain how to find an augmentation in $G'$ in $O(D + \sqrt{n})$ rounds.
In {\unAug}, when we find an augmentation in $G'$, we scan $T$ from the leaves to the root, and whenever we get to a tree edge that is still not covered we cover it by the maximal edge possible. An edge $e$ is the maximal edge between $e=\{u,w\}$ and $e'=\{u',w'\}$, where $u,u'$ are ancestors of $w,w'$ respectively, if and only if $u$ is an ancestor of $u'$.

We define a variant of this algorithm, {\unAugTag}, whose input is the tree $T$, the graph $G$, and a set $T_0$ of tree edges from $T$ that are already covered. {\unAugTag} finds an augmentation in $G$ by applying {\unAug}, with the difference that now we cover only the tree edges that are not in $T_0$. 
When we cover an edge, we still cover it by the maximal edge possible. 

The general structure of the algorithm for finding an augmentation in $G'$ is as follows:
\begin{itemize}
\item Each leaf $v$ covers the tree edge $\{v,p(v)\}$ by the maximal edge possible.
\item We cover global edges that are still not covered by applying {\unAugTag} on $T_F$.
\item We cover local edges that are still not covered by applying {\unAugTag} on each fragment separately.
\end{itemize}

We next describe how to implement the above efficiently in a distributed way.

\subsubsection{Covering Leaf Edges}

For a leaf $v$ in $T$, we say that the tree edge $\{v,p(v)\}$ is a \textit{leaf edge}.
 
We start the algorithm by covering leaf edges: each leaf $v$ covers the tree edge $\{v,p(v)\}$ by the maximal edge possible. Since each vertex knows the labels of all the edges incoming to it in $G'$, it knows which is the maximal one as in {\unAug}. This computation does not require any communication. However, for the rest of the algorithm each vertex $u$ needs to know if the tree edge $\{u,p(u)\}$ is covered by one of the edges we added in order to cover leaf edges. In order to do that, we need coordination between the vertices. We divide this task into a local coordination at each fragment, and a global coordination between fragments.

\paragraph*{Local coordination:$\ $}
In this part, each vertex $v$ learns about the maximal edge that covers $\{v,p(v)\}$ among edges added to the augmentation by leaves in its fragment, if such exists.

In order to do this, we apply the following algorithm in each fragment separately: we scan the fragment from its leaves to its root by having each leaf $v$ of the fragment that is also a leaf in $T$ send to its parent the labels of the edge it added. Any leaf of the fragment that is not a leaf in $T$ sends to its parent an empty message. 

Each internal vertex $v$ gets messages from all its children. If at least one of the messages is an edge that covers $\{v,p(v)\}$, $v$ sends to its parent the labels of the maximal edge among those it received from its children. Otherwise, it sends an empty message. Note that using the labels of an edge $e=\{u,w\}$, where $u$ is an ancestor of $w$, a vertex $v$ knows if this edge covers $\{v,p(v)\}$ using LCA computations: it checks if $v$ is an ancestor of $w$ and if $u$ is an ancestor of $p(v)$. It can also learn which is the maximal edge by LCA computations.

Note that by the end of the algorithm each vertex $v$ learns if the tree edge $\{v,p(v)\}$ is covered by an edge that one of the leaves of the fragment adds to the augmentation, and the root of the fragment, $v'$, learns the labels of the maximal edge added to the augmentation by leaves of the fragment that covers the global edge $\{v',p(v')\}$, if such exists.
The round complexity of this part is proportional to the diameter of the fragment, and is bounded by $O(\sqrt{n})$.

\paragraph*{Global coordination:$\ $} 
Each vertex $v$ that is a root of a fragment, excluding $r$, sends over the BFS tree the labels of the maximal edge added to the augmentation by leaves of the fragment that covers $\{v,p(v\}$, if such exists. Since there are at most $O(\sqrt{n})$ fragments, there are at most $O(\sqrt{n})$ messages sent. So we can distribute these messages over the BFS tree to all the vertices in $O(D + \sqrt{n})$ rounds. 

Note that using the labels of an edge $e$, a vertex $v$ knows if this edge covers $\{v,p(v)\}$ using LCA computations. In particular, each vertex $v$ knows if the tree edge $\{v,p(v)\}$ is covered by one of the $O(\sqrt{n})$ edges sent to all the vertices. 

Note that although there may be $\omega(\sqrt{n})$ leaves in $T$, and each one adds an edge to the augmentation, after the local coordination and the global coordination, in which each vertex receives information about $O(\sqrt{n})$ edges, each vertex $v$ knows if the tree edge $\{v,p(v)\}$ is covered by an edge added by a leaf in $T$. This is proven in the next claim.

\begin{claim} \label{glob}
After the local and global coordination, each vertex $v$ knows if the tree edge $\{v,p(v)\}$ is covered by an edge added by a leaf in $T$. 
\end{claim}

\begin{proof}
Let $v$ be a vertex and assume there is an edge $e=\{u,w\}$ added by a leaf $u$ in $T$, which covers the tree edge $\{v,p(v)\}$. If $u$ is in the same fragment as $v$, in the local coordination $v$ learns about the maximal edge added by a leaf in the fragment that covers $\{v,p(v)\}$, and in particular it learns that there is an edge that covers $\{v,p(v)\}$, as needed. Assume now that $u,v$ are in different fragments, $F_u, F_v$, and there is no leaf in $F_v$ that adds an edge that covers $\{v,p(v)\}$. Let $r_u$ be the root of $F_u$, and let $e_u$ be the edge $r_u$ sends over the BFS tree. Note that $e_u$ covers $\{v,p(v)\}$ because the edge $e$ covers $\{r_u,p(r_u)\}$ and covers $\{v,p(v)\}$, and $e_u$ is the maximal edge that covers $\{r_u,p(r_u)\}$. So $v$ learns there is an edge added by a leaf that covers $\{v,p(v)\}$, as needed. 
\end{proof}

\begin{claim} \label{glob2}
After the global coordination, each vertex knows if a global edge $\{v,p(v)\}$ is covered by an edge added by a leaf in $T$. 
\end{claim}

\begin{proof}
Note that all the vertices know the labels of all the global edges. If a global edge $\{v,p(v)\}$ is covered by an edge $\{u,w\}$, where $u$ is a leaf and $u$ is in the fragment $F_u$, then the edge $e_u$ sent by the root $r_u$ of $F_u$ covers $\{v,p(v)\}$ as well. Since all vertices learn about the labels of $e_u$, by LCA computations they can learn that there is an edge added by a leaf that covers $\{v,p(v)\}$.
\end{proof}

\subsubsection{Covering Global Edges}

The goal now is to cover global edges that are still not covered by applying {\unAugTag} to $T_F$. Note that the maximal edge that covers a global edge must be a maximal edge incoming to a fragment: assume that $e=\{v_{F_1},v_{F_2} \}$ is the maximal edge that covers the global edge $e'$ in $T_F$ and that $e$ is incoming to $F_1$, then the maximal edge $e_M$ incoming to $F_1$ covers $e'$ as well. Since $e$ is the maximal edge that covers $e'$, it follows that $e=e_M$.
Therefore, in order to apply {\unAugTag} it is enough to know the maximal incoming edge to each fragment in $T_F$ (they are the only edges that may be added to the augmentation), and which global edges are already covered. Note that all the vertices know which global edges are already covered after the global coordination, according to Claim \ref{glob2}. 

In order to learn the maximal edge incoming to a fragment, each fragment computes this edge by scanning the fragment from its leaves to its root. A leaf sends to its parent the labels of the maximal edge incoming to it.
Each internal vertex $v$, excluding the root of the fragment, sends to its parent the labels of the maximal edge covering $\{v,p(v)\}$ among the edges it receives from its children and the maximal edge incoming to it (it can compute the maximal edge by LCA computations using the labels of the edges). At the end, the root $v$ of each fragment learns about the maximal edge incoming to the fragment that covers the global edge $\{v,p(v)\}$, if such exists. 

The root of each fragment (excluding $r$) distributes over the BFS tree the (local and global) labels of the maximal edge $e$ incoming to its fragment. Note that the global labels of $e$ indicate which fragments are connected by $e$. Since there are $O(\sqrt{n})$ fragments, we can distribute all this information over the BFS tree to all the vertices in $O(D+\sqrt{n})$ rounds. 

The computation at each fragment takes $O(\sqrt{n})$ rounds and the communication between fragments takes $O(D + \sqrt{n})$ rounds. So the overall time complexity of this part is bounded by $O(D + \sqrt{n})$ rounds.

After all the vertices learn the maximal edge incoming to each fragment and which global edges are already covered, each vertex can apply {\unAugTag} on $T_F$ locally, without any communication. When a vertex covers a global edge, it covers it by the maximal edge possible with respect to $T$. This is also a maximal edge with respect to $T_F$, but there may be several edges in $T_F$ that connect the same fragments, in which case we use the local labels in order to choose the maximal between them.
Note that after applying {\unAugTag}, each vertex knows which of the maximal edges incoming to a fragment is added to the augmentation and, in particular, a vertex $v$ knows if the maximal edge incoming to it is added to the augmentation and if there is an edge added to the augmentation that covers the tree edge $\{v,p(v)\}$. The edges added to the augmentation cover all the global edges and some of the local edges. 

We next cover the local edges that are still not covered.

\subsubsection{Covering Local Edges}

In this part, we cover local edges that are still not covered by applying {\unAugTag} locally in each fragment. The idea is to scan the fragment from its leaves to its root, and each time we get to a tree edge that is still not covered, we cover it by the maximal edge possible. 

Note that the maximal edge covering a tree edge $\{v,p(v)\}$ may be the maximal edge incoming to any vertex in the subtree rooted at $v$. In particular, it may be incoming to a vertex in another fragment $F$. However, in this case it must be the maximal edge incoming to $F$. Since each vertex knows the maximal edges incoming to each fragment, we can compute the maximal edge covering a tree edge without communication with other fragments. Note that we also know which edges are already covered by edges already added to the augmentation. Denote by $T_0$ all the tree edges that are covered by edges added to the augmentation in order to cover leaf edges or global edges. 

The distributed implementation of {\unAugTag} is very similar to {\unAug}. However, there are several differences:
\begin{itemize}
\item We cover only tree edges that are not in $T_0$. Note that each vertex $v$ knows if the edge $\{v,p(v)\}$ is in $T_0$.
\item In order to apply the algorithm we need to compute for each edge the maximal edge that covers it. A leaf $v$ of the fragment computes this edge among the edges incoming to it and the maximal edges incoming to a fragment. An internal vertex computes it as in {\unAug}, using the edges it receives from its children and the edges incoming to it.
\item At the end of the algorithm, as in {\unAug}, each vertex knows if the maximal edge it sent to its parent is added to the augmentation. In particular, each vertex in the fragment learns if the maximal edge incoming to it is added to the augmentation by another vertex in the fragment. However, we may decide to add to the augmentation edges incoming to other fragments. We explain next how to distribute this information between fragments.
\end{itemize}

The computation on each fragment takes $O(\sqrt{n})$ rounds. In order to end the algorithm, each vertex needs to know if the maximal edge incoming to it is added to the augmentation, which we achieve using global coordination between the fragments.

\paragraph*{Global coordination:$\ $} 
Note that when we apply {\unAugTag}, a vertex may decide to add to the augmentation one of the $O(\sqrt{n})$ maximal edges incoming to a fragment. A vertex that decides to add such an edge sends the labels of this edge over the BFS tree. Since there are at most $O(\sqrt{n})$ such edges, there are at most $O(\sqrt{n})$ different messages sent over the BFS tree, and we can distribute this information over the BFS tree to all the vertices in $O(D + \sqrt{n})$ rounds. So, at the end each vertex knows if the maximal edge incoming to it is added to the augmentation as needed. 

Note that we covered all the edges that were still not covered, so the solution obtained is an augmentation. 
The overall time complexity of the algorithm for finding an augmentation in $G'$ is $O(D + \sqrt{n})$ rounds.

We next show that it is a 2-approximation to the optimal augmentation in $G'$. As in {\unTAP}, after we have an augmentation in $G'$ we can convert it to an augmentation in $G$ that is at most twice the size, which implies that we get a 4-approximation to the optimal augmentation in $G$. 

\begin{lemma} \label{time4}
The time complexity of the whole algorithm is $O(D+\sqrt{n}\log^*{n})$ rounds.
\end{lemma}

\begin{proof}
Breaking the tree $T$ into fragments takes $O(D+\sqrt{n}\log^*{n})$ rounds, using the MST algorithm of Kutten and Peleg \cite{kutten1995fast}. Building $G'$ from $G$ takes $O(D+\sqrt{n})$ rounds by Lemma \ref{timeb}, Finding an augmentation in $G'$ takes $O(D+\sqrt{n})$ rounds, as discussed throughout.
\end{proof}

\subsection{Approximation Ratio} \label{sec:app_approx}

\subsubsection*{Intuition for the analysis} 

We next show that the size of our solution is at most twice the size of an optimal augmentation in $G'$. Denote by $A$ the solution obtained by the algorithm and by $A^*$ an optimal augmentation in $G'$. 
In the correctness proof of {\unTAP} we show a one-to-one mapping from $A$ to $A^*$, but this mapping is no longer one to one here. However, if we could show that each edge in $A^*$ is mapped to by at most two edges from $A$, we can obtain a 2-approximation. Unfortunately, this does not hold either. 

Our approach is to divide the edges in $A$ to two types $A_1,A_2$ as follows. We map each edge $e \in A$ to a corresponding path $P_e$ in $T$. If $P_e$ contains an internal vertex with more than one child in the tree we say that $e \in A_1$, otherwise $e \in A_2$. Then, we show that $|A_1| \leq 2|A^*|$, and $|A_2| \leq 2|A^*|$. The main idea is that the number of edges in $A_1$ is related to the degrees of internal vertices in $T$, which affects the number of leaves in the tree. We use this in order to show that $|A_1| \leq 2\ell$ where $\ell$ is the number of leaves in $T$. Note that $\ell$ is a lower bound on the size of any augmentation in $G'$, since we need to add to the augmentation a different edge in order to cover each one of the leaves. This gives $|A_1| \leq 2|A^*|$.

In order to show that $|A_2| \leq 2|A^*|$, we use the fact that the edges of $A_2$ correspond to tree paths with a simple structure. This allows us to show a mapping between $A_2$ to $A^*$ in which each edge in $A^*$ is mapped to by at most two edges from $A_2$, giving $|A_2| \leq 2|A^*|$. 

In conclusion, $|A| = |A_1 \cup A_2| \leq 4|A^*|$. A more delicate analysis extending these ideas gives $|A| \leq 2|A^*|$. This gives a 2-approximation to the optimal augmentation of $G'$, which results in a 4-approximation to the optimal augmentation in $G$. 



\subsubsection*{Approximation ratio analysis}

Each edge $e \in A$ is added to $A$ in the algorithm in order to cover some tree edge that is still not covered, denote this edge by $t(e)$. Let $t(A)$ be all such tree edges. For each edge $t(e) \in t(A)$, we go up in the tree until we get to another tree edge $t' \in t(A)$, or to the root in case there is no such edge. If $t'$ exists, we denote it by $t_2(e)$.

\begin{claim} \label{obs1}
Let $e \in A$, such that $t(e)$ is a global edge and $t_2(e)$ exists. Then there is no edge that covers both $t(e)$ and $t_2(e)$. 
\end{claim}

\begin{proof}
Note that $t(e),t_2(e) \in t(A)$, i.e., when we get to them in the algorithm they are still not covered. Also, since $t(e)$ is a global edge and $t_2(e)$ is on the path between $t(e)$ to $r$, we get to $t(e)$ before $t_2(e)$ in the algorithm. When we get to $t(e)$ in the algorithm we cover it by the maximal edge possible. This edge does not cover $t_2(e)$, otherwise $t_2(e) \not \in t(A)$.
\end{proof}

Let $V_{>1}$ be the set of vertices with more than one child in $T$. We write $A=A_1 \cup A_2$ in the following way: let $e \in A$, $t(e)=\{v,p(v)\}$, and $t_2(e)=\{u,p(u)\}$ if it exists. Let $P(e)=\{p(v)=v_1,...,v_k\}$ be the vertices on the tree path between $p(v)$ and $v_k=u$ if $t_2(e)$ exists, or between $p(v)$ and $v_k=r$ otherwise. If there is a vertex $v' \in P(e)$ such that $v' \in V_{>1}$, we say that $e \in A_1$, and otherwise $e \in A_2$. 

\begin{claim} \label{obs2}
There is at most one edge $e \in A_2$ such that $t_2(e)$ does not exist.
\end{claim}

\begin{proof}
Assume there are two edges $e_1,e_2 \in A_2$ such that $t_2(e_1),t_2(e_2)$ does not exist. Then, on the path $P_1$ between $t(e_1)$ to $r$ and on the path $P_2$ between $t(e_2)$ and $r$ there is no vertex in $V_{>1}$. It can only happen if one of $P_1,P_2$ is contained in the other. Assume without loss of generality that $P_1$ contains $P_2$. But then on the path $P_1$ there is another edge in $t(A)$, so $t_2(e_1)$ exists.
\end{proof}

\begin{claim} \label{path}
Let $t=\{v,p(v)\} \in t(A)$ and $e \in A_2$ such that $t_2(e)=\{u,p(u)\}$ where $u$ is an ancestor of $p(v)$. Then $t(e)$ is on the tree path $P'=\{v,p(v),...,u\}$ between $t$ to $t_2(e)$. 
\end{claim}

\begin{proof}
If $t=t(e)$ we are done. Note that since $e \in A_2$, on the tree path $P(e)$ between $t(e)$ to $t_2(e)$ there are no vertices in $V_{>1}$ and no other edges in $t(A)$. If $t(e)$ is not on the path $P'$ between $t$ to $t_2(e)$, it follows that there is a vertex $v' \in P(e)$ such that $v' \in V_{>1}$, at the point where $P(e)$ and $P'$ diverge, or $t \in t(A)$ is in $P(e)$. Either case gives a contradiction.
\end{proof}

Let $A_1^*$ be the edges in $A^*$ that cover leaf edges. Let $\ell$ be the number of leaves in $T$. 

\begin{claim}
$|A_1^*| = \ell$.
\end{claim}

\begin{proof}
Each leaf edge is covered by a different edge in $A^*$ since all the edges in $G'$ are between an ancestor and its descendant in the tree. Also, each leaf edge is covered by exactly one edge in $A^*$, because if there are two edges $e_1,e_2$ that cover the same leaf edge, and assume without loss of generality that $e_1$ is the maximal between them, then it covers all the edges covered by $e_2$, which contradicts the optimality of $A^*$.
\end{proof}

We divide the leaves to two types in the following way: we map each leaf $v$ to the edge $e_v \in A_1^*$ that covers the corresponding leaf edge. For each edge $e_v$ in $A_1^*$ we look at the corresponding path of tree edges that it covers. If one of the vertices in this path is in $V_{>1}$ we say that $v \in L_1$, otherwise $v \in L_2$. Let $\ell_1=|L_1|$ and $\ell_2=|L_2|$, giving $\ell=\ell_1+\ell_2$.

\begin{claim}
If there is an edge in $A_1^*$ of the form $\{v,r\}$ that covers a leaf $v \in L_2$ then the solution given by the algorithm is optimal.
\end{claim}

\begin{proof}
Note that if there is an edge in $A_1^*$ of the form $\{v,r\}$ that covers a leaf $v \in L_2$ it follows that $T$ is just the path between $v$ to $r$ and there is one edge that covers it. In such a case, our algorithm is optimal because it starts by adding the maximal edges that cover leaves, and hence it adds this edge and no other edge.
\end{proof}

We next assume that there are no edges in $A_1^*$ of the form $\{v,r\}$ that cover a leaf $v \in L_2$.
According to our assumption, each edge in $A_1^*$ that covers a leaf edge $e_v$ such that $v \in L_2$ is of the form $\{v,u\}$ where $u \neq r$. There are exactly $\ell_2$ tree edges of the form $\{u,p(u)\}$ for all such veritces $u$, denote them by $E_2$. Let $A_2^*$ be all the edges in $A^*$ that cover edges in $E_2$.

\begin{claim}
$A_1^* \cap A_2^* = \emptyset$.
\end{claim}

\begin{proof}
Let $e = \{u,p(u)\} \in E_2$, so there is a leaf $v \in L_2$ such that $\{v,u\} \in A_1^*$.
Note that $e$ is not covered by edges from $A_1^*$ because by the definition of $L_2$, the subtree rooted at $u$ is the path between $v$ to $u$, and the only edge from $A_1^*$ that covers edges on this path is $\{v,u\}$, which does not cover $\{u,p(u)\}$.
\end{proof}

\begin{claim} \label{A2_size}
$|A_2^*| \geq \ell_2$.
\end{claim}

\begin{proof}
There are exactly $\ell_2$ edges in $E_2$. We show that each of them is covered by a different edge from $A_2^*$.
Note that if $\{u,p(u)\} \in E_2$ then the subtree rooted at $u$ is a path, in which all edges are covered by an edge from $A_1^*$. In particular, in this path there are no other tree edges from $E_2$. It follows that edges in $E_2$ cannot be on the same path between a leaf and $r$ in the tree, and cannot be covered by the same edge because all the edges in $G'$ are between an ancestor to its descendant in the tree. The claim follows.
\end{proof}

Let $A_3^* = A^* \setminus (A_1^* \cup A_2^*)$.
In order to show that $|A| \leq 2|A^*|$, we prove the following two lemmas:

\begin{lemma} \label{lem1}
$|A_1| \leq 2|A_1^*| -2$.
\end{lemma}

\begin{lemma} \label{lem2}
$|A_2| \leq 2|A_2^*|+2|A_3^*|+1$.
\end{lemma}

To prove Lemma \ref{lem1}, we map edges in $A_1$ to vertices in $V_{>1}$ in the following way:
Let $e \in A_1$, such that $t(e)=\{v,p(v)\}$. By definition of $A_1$, on the path $P(e)$ there is a vertex in $V_{>1}$. We map $e$ to a vertex $u \in V_{>1}$ that is closest to $v$ on this path. We need the following claim.

\begin{claim} \label{claim_V1}
If $u \in V_{>1}$ has $k$ children then it is mapped to by at most $k$ edges. 
\end{claim}

\begin{proof}
The edges $e$ mapped to $u$ are such that $t(e)$ is in the subtree rooted at $u$. We divide this subtree to $k$ parts according to its children. 
Let $u'$ be a child of $u$, let $T_{u'}$ be the subtree rooted at $u'$, and let $T' = T_{u'} \cup \{u,u'\}$. We show that there is at most one edge $e$ where $t(e) \in T'$ that is mapped to $u$. Assume there are 2 edges $e_1,e_2 \in A_1$ such that $t(e_1),t(e_2) \in T'$ that are mapped to $u$. Let $P_1,P_2$ be the paths between $t(e_1)$ and $u$, and between $t(e_2)$ and $u$ respectively. 

If one of $P_1,P_2$ is contained in the other, and assume without loss of generality that $P_1$ contains $P_2$, then $t(e_2)$ is on the path between $t(e_1)$ to $u$. From the definition of $A_1$ there is a vertex $v' \in V_{>1}$ between $t(e_1)$ to $t(e_2)$, which is closer to $t(e_1)$ than $u$, a contradiction to the fact that $e_1$ is mapped to $u$. 
Otherwise, $P_1$ and $P_2$ diverge in some vertex $v'$ in $T_{u'}$, but then $v'$ is a vertex in $V_{>1}$ that is closer to $t(e_1)$ and $t(e_2)$, a contradiction.  
\end{proof}

Using Claim \ref{claim_V1}, we prove Lemma \ref{lem1}.

\begin{proof} [Proof of Lemma \ref{lem1}]
For each internal vertex (including $r$) we choose one child and call it the \emph{main child}, and we call the other children \emph{extra children}. Note that all the vertices in $T$ except $r$ are children of some parent, so there are $n-1$ children in $T$. Denote by $x$ the number of extra children in $T$. There are $n-\ell$ internal vertices, so there are $n-\ell$ main children, giving $x=n-1-(n-\ell)=\ell-1$.

\remove{
We map edges in $A_1$ to vertices in $V_{>1}$ in the following way:
Let $e \in A_1$, such that $t(e)=\{v,p(v)\}$. By definition of $A_1$, on the path $P(e)$ there is a vertex in $V_{>1}$. We map $e$ to a vertex $u \in V_{>1}$ that is closest to $v$ on this path. 

\begin{claim}
If $u \in V_{>1}$ has $k$ children then it is mapped to by at most $k$ edges. 
\end{claim}

\begin{proof}
The edges $e$ mapped to $u$ are such that $t(e)$ is in the subtree rooted at $u$. We divide this subtree to $k$ parts according to its children. 
Let $u'$ be a child of $u$, let $T_{u'}$ be the subtree rooted at $u'$, and let $T' = T_{u'} \cup \{u,u'\}$. We show that there is at most one edge $e$ where $t(e) \in T'$ that is mapped to $u$. Assume there are 2 edges $e_1,e_2 \in A_1$ such that $t(e_1),t(e_2) \in T'$ that are mapped to $u$. Let $P_1,P_2$ be the paths between $t(e_1)$ and $u$, and between $t(e_2)$ and $u$ respectively. 

If one of $P_1,P_2$ is contained in the other, and assume without loss of generality that $P_1$ contains $P_2$, then $t(e_2)$ is on the path between $t(e_1)$ to $u$. From the definition of $A_1$ there is a vertex $v' \in V_{>1}$ between $t(e_1)$ to $t(e_2)$, which is closer to $t(e_1)$ than $u$, a contradiction to the fact that $e_1$ is mapped to $u$. 
Otherwise, $P_1$ and $P_2$ diverge in some vertex $v'$ in $T_{u'}$, but then $v'$ is a vertex in $V_{>1}$ that is closer to $t(e_1)$ and $t(e_2)$, a contradiction.  
\end{proof}
}

By Claim \ref{claim_V1}, if $u \in V_{>1}$ has $k$ children then it is mapped to by at most $k$ edges. 
It follows that if $u$ has $k-1$ extra children, we map to it at most $k$ edges from $A_1$. In the worst case, the number of edges in $A_1$ is twice the number of extra children. In conclusion, $|A_1| \leq 2x=2\ell-2=2|A_1^*|-2$, which completes the proof.
\end{proof}

We next prove Lemma \ref{lem2}. According to Claim \ref{obs2}, there is at most one edge $e' \in A_2$ such that $t_2(e')$ does not exist. For the proof of Lemma \ref{lem2}, we map all the edges in $A_2$ except $e'$ to edges in $A^*$ in the following way: let $e \in A_2$. If $t(e)$ is a leaf edge or a local edge, we map $e$ to an edge in $A^*$ that covers $t(e)$. Otherwise, we map $e$ to an edge in $A^*$ that covers $t_2(e)$. We need the following claims.

\begin{claim} \label{global_claim}
If $t(e_1),t(e_2)$ are both global edges that are not leaf edges then $e_1,e_2$ cannot be mapped to the same edge $e \in A^*$.
\end{claim}

\begin{proof}
Assume that $t(e_1),t(e_2)$ are both global edges that are not leaf edges. In this case, $e$ covers both $t_2(e_1),t_2(e_2)$, so the path $P$ between them is between a descendant to its ancestor in the tree. Assume without loss of generality that $t_2(e_2)$ is closer to the root in $P$. By Claim \ref{path}, $t(e_2)$ is in $P$, and it follows that $e$ covers $t(e_2),t_2(e_2)$ where $t(e_2)$ is a global edge, a contradiction to Claim \ref{obs1}.   
\end{proof}

\begin{claim} \label{local}
If $t(e_1),t(e_2)$ are both local or leaf edges then $e_1,e_2$ cannot be mapped to the same edge $e \in A^*$.
\end{claim}

\begin{proof}
Assume that $t(e_1),t(e_2)$ are both local or leaf edges. In this case, $e$ covers both $t(e_1),t(e_2)$. Assume without loss of generality that $t(e_2)$ is closer to the root. Note that $t(e_1),t(e_2)$ cannot be in the same fragment, and $t(e_1)$ cannot be a leaf edge, because when we get to $t(e_1)$ in the algorithm, we cover it by the maximal edge possible, which covers $t(e_2)$ because the edge $e$ covers $t(e_1)$ and $t(e_2)$. If $t(e_1)$ is a leaf edge or is in the same fragment as $t(e_2)$, it follows that $t(e_2) \not \in t(A)$. The same argument shows that $t(e_1),t_2(e_1)$ are not in the same fragment ($t_2(e_1)$ is on the path between $t(e_1)$ and $t(e_2)$ and is covered by $e$ also). 

Hence, there is a global edge on the path $P$ between $t(e_1)$ and $t_2(e_1)$. Let $g$ be a global edge in $P$ that is closest to $t(e_1)$. If $g \in t(A)$, then when we get to $g$ in the algorithm it is still not covered, and we add the maximal edge possible in order to cover it. This edge covers $t_2(e_1)$ because the edge $e$ covers both $g$ and $t_2(e_1)$. This contradicts the fact that $t_2(e_1) \in t(A)$. 

Hence, $g \not \in t(A)$, and when we get to it in the algorithm it is already covered by an edge $\widetilde{e}$ added in order to cover a tree edge $g'$. The edge $g'$ may be a leaf edge or a global edge, so $g' \neq t(e_1)$. Note that $t(e_1)$ is on the path $P$ between $g'$ and $g$, otherwise we have a vertex in $V_{>1}$ on the path $P'$ between $t(e_1)$ to $g$ (and in particular between $t(e_1)$ and $t_2(e_1)$) at the point where $P$ and $P'$ diverge, or another global edge $g'$ between $t(e_1)$ and $g$ (if $g'$ is a leaf edge it cannot be on the path between $t(e_1)$ and $g$). Either case gives a contradiction. Hence, $\widetilde{e}$ covers $t(e_1)$, but then $t(e_1) \not \in t(A)$.    
\end{proof}

\begin{proof} [Proof of Lemma \ref{lem2}]
Our proof is based on the following claims:

\begin{enumerate}[(I)]
\item There are at most $\ell_2$ edges in $A_2$ that are mapped to edges in $A_1^*$. \label{I}
\item Each edge in $A^*$ is mapped to by at most two edges. \label{II}
\item \label{III} Each edge in $A_2^*$ is mapped to by at most one edge. 
\end{enumerate}

From the above three claims we get that the number of edges in $A_2$ is bounded by $1 + \ell_2 + |A_2^*| + 2|A_3^*|$ as follows: there is at most one edge $e'$ that is not included in the mapping, there are at most $\ell_2$ edges in $A_2$ that are mapped to edges in $A_1^*$, at most $|A_2^*|$ edges that are mapped to edges in $A_2^*$, and at most $2|A_3^*|$ edges that are mapped to edges in $A_3^*$. Note that by Claim \ref{A2_size}, $|A_2^*| \geq \ell_2$. It follows that $|A_2| \leq 2|A_2^*| + 2|A_3^*| + 1$ as needed.

\begin{proof} [Proof of (\ref{I})]
Let $e^* \in A_1^*$. Then $e^*$ covers a leaf edge $t=\{v,p(v\}$. Let $P$ be the path of tree edges that $e^*$ covers. Note that $t$ is the only edge in $P$ such that $t \in t(A)$: since we start the algorithm by covering each leaf edge by the maximal edge possible, then if $e \in A$ is added in the algorithm in order to cover $t$, it covers also all the edges in $P$. Since the only edges that may be mapped to $e^*$ are edges $\widetilde{e}$ such that $t(\widetilde{e})$ or $t_2(\widetilde{e})$ are in $P$, it follows that the only edge in $A$ that may be mapped to $e^*$ is the edge $e$. 
Note that if $e \in A_2$, in the path $P(e)$ there are no vertices in $V_{>1}$, it follows that in $P$ there are no vertices in $V_{>1}$, so $v \in L_2$ by definition. It follows that there are at most $\ell_2$ edges in $A_2$ that are mapped to edges in $A_1^*$.
\end{proof}

\begin{proof} [Proof of (\ref{II})]

Assume that there are two edges $e_1,e_2$ in $A_2$ that are mapped to the same edge $e \in A^*$. 
\remove{
\begin{claim}
If $t(e_1),t(e_2)$ are both global edges that are not leaf edges then $e_1,e_2$ cannot be mapped to the same edge $e \in A^*$.
\end{claim}

\begin{proof}
Assume that $t(e_1),t(e_2)$ are both global edges that are not leaf edges. In this case, $e$ covers both $t_2(e_1),t_2(e_2)$, so the path $P$ between them is between a descendant to its ancestor in the tree. Assume without loss of generality that $t_2(e_2)$ is closer to the root in $P$. By Claim \ref{path}, $t(e_2)$ is in $P$, and it follows that $e$ covers $t(e_2),t_2(e_2)$ where $t(e_2)$ is a global edge, a contradiction to Claim \ref{obs1}.   
\end{proof}

\begin{claim} \label{local}
If $t(e_1),t(e_2)$ are both local or leaf edges then $e_1,e_2$ cannot be mapped to the same edge $e \in A^*$.
\end{claim}

\begin{proof}
Assume that $t(e_1),t(e_2)$ are both local or leaf edges. In this case, $e$ covers both $t(e_1),t(e_2)$. Assume without loss of generality that $t(e_2)$ is closer to the root. Note that $t(e_1),t(e_2)$ cannot be in the same fragment, and $t(e_1)$ cannot be a leaf edge, because when we get to $t(e_1)$ in the algorithm, we cover it by the maximal edge possible, which covers $t(e_2)$ because the edge $e$ covers $t(e_1)$ and $t(e_2)$. If $t(e_1)$ is a leaf edge or is in the same fragment as $t(e_2)$, it follows that $t(e_2) \not \in t(A)$. The same argument shows that $t(e_1),t_2(e_1)$ are not in the same fragment ($t_2(e_1)$ is on the path between $t(e_1)$ and $t(e_2)$ and is covered by $e$ also). 

Hence, there is a global edge on the path $P$ between $t(e_1)$ and $t_2(e_1)$. Let $g$ be a global edge in $P$ that is closest to $t(e_1)$. If $g \in t(A)$, then when we get to $g$ in the algorithm it is still not covered, and we add the maximal edge possible in order to cover it. This edge covers $t_2(e_1)$ because the edge $e$ covers both $g$ and $t_2(e_1)$. This contradicts the fact that $t_2(e_1) \in t(A)$. 

Hence, $g \not \in t(A)$, and when we get to it in the algorithm it is already covered by an edge $\widetilde{e}$ added in order to cover a tree edge $g'$. The edge $g'$ may be a leaf edge or a global edge, so $g' \neq t(e_1)$. Note that $t(e_1)$ is on the path $P$ between $g'$ and $g$, otherwise we have a vertex in $V_{>1}$ on the path $P'$ between $t(e_1)$ to $g$ (and in particular between $t(e_1)$ and $t_2(e_1)$) at the point where $P$ and $P'$ diverge, or another global edge $g'$ between $t(e_1)$ and $g$ (if $g'$ is a leaf edge it cannot be on the path between $t(e_1)$ and $g$). Either case gives a contradiction. Hence, $\widetilde{e}$ covers $t(e_1)$, but then $t(e_1) \not \in t(A)$.    
\end{proof}
}
From Claim \ref{global_claim}, if $t(e_1),t(e_2)$ are both global edges that are not leaf edges then $e_1,e_2$ cannot be mapped to the same edge $e \in A^*$. From Claim \ref{local}, if $t(e_1),t(e_2)$ are both local or leaf edges then $e_1,e_2$ cannot be mapped to the same edge $e \in A^*$.
It follows that there is no edge in $A^*$ that is mapped to by three or more edges. Assume there are three edges $e_1,e_2,e_3$ that are mapped to the same edge $e \in A^*$. At least two of $t(e_1),t(e_2),t(e_3)$ are local or leaf edges, or at least two of them are global edges that are not leaf edges. Either case gives a contradiction. It follows that each edge in $A^*$ is mapped to by at most two edges as needed.
\end{proof}

\begin{proof} [Proof of (\ref{III})]
Let $e^* \in A_2^*$, and let $P$ be the path of tree edges covered by $e^*$. By definition, there is an edge $t=\{u,p(u)\} \in E_2$ that is covered by $e^*$, and the subtree $T_u$ rooted at $u$ is a path which is covered by an edge $e_1=\{u,v\} \in A_1^*$ where $v$ is a leaf. Note that there is only one edge $e_2$ such that $t(e_2) \in T_u$, which is the edge $e_2$ that covers $\{v,p(v)\}$, since all other edges in $T_u$ are already covered by $e_2$ (it is the maximal edge possible, and in particular covers all tree edges covered by $e_1$), and $e_2$ is mapped to $e_1 \not \in A_2^*$.

The only edges that may be mapped to $e^*$ are edges $e$ such that $t(e)$ or $t_2(e)$ are in $P$. There may be at most one edge $e$ mapped to $e^*$ such that $t(e)$ is a local edge, according to Claim \ref{local}. So, if there is another edge $e_3$ mapped to $e^*$ it must be a global edge such that $t_2(e_3)$ is in $P$. Note that $t(e_3)$ cannot be in $P$, otherwise we have a contradiction to Claim \ref{obs1}, and it cannot be in $T_u$ as explained above. Let $v'$ be the first vertex in $P$ on the path $P(e_3)$ between $t(e_3)$ to $t_2(e_3)$. Note that $v' \in V_{>1}$, since it has a child not in $P$ on the path $P(e)$ and another child in $P$ because it is an ancestor of $u$. In such a case $e_3$ cannot be in $A_2$. Hence, there is at most one edge in $A_2$ that is mapped to each edge in $A_2^*$ as needed. 
\end{proof}
This completes the proof of Lemma \ref{lem2}. 
\end{proof}

\uTAPtwo*

\begin{proof}
By Lemma \ref{lem1} and Lemma \ref{lem2}, we have:
$$|A|=|A_1 \cup A_2| \leq 2|A_1^*|+2|A_2^*|+2|A_3^*| - 1 \leq 2|A^*|.$$
Hence, $A$ is an augmentation in $G'$, whose size is at most twice the size of an optimal augmentation in $G'$. It corresponds to an augmentation in $G$ whose size is at most 4 times the size of an optimal augmentation in $G$ according to Lemma \ref{corr}. The running time is $O(D+\sqrt{n}\log^*{n})$ rounds by Lemma \ref{time4}.
\end{proof}

\section{Lower Bounds} \label{sec:lower}

\subsection{An $\Omega(D)$ Lower Bound for TAP in the LOCAL model} \label{sec:app_lproof}

We show that TAP is a global problem, which admits a lower bound of $\Omega(D)$ rounds, even in the LOCAL model where the size of messages is unbounded. In the LOCAL model, a vertex can learn in $r$ rounds its $r$-neighborhood, which consists of all the vertices and edges at distance at most $r$ from it. In addition, if the $r$-neighborhood of a vertex is the same in two different graphs it cannot distinguish between them in any algorithm that takes at most $r$ rounds. Based on this, we show the following.

\local*

\begin{proof}
Let $k$ be an even integer, and consider the graph $G_1$ that consists of a path $P$ of $n=2k+1$ vertices $\{v_0,v_1,...,v_{2k}\}$, and the additional edges $\{v_{2i},v_{2(i+1)}\}$ for $0 \leq i < k$. Consider also the graph $G_2=G_1 \cup \{v_0,v_{2k}\}$. Both graphs have diameter $D=\Theta(k)$. Consider an instance for TAP where $T$ is the path $P$ for both graphs $G_1$ and $G_2$. It is easy to verify that an optimal augmentation in $G_1$ includes all the edges $\{v_{2i},v_{2(i+1)}\}$ for $0 \leq i < k$, as this is the only way to cover all the edges. However, in $G_2$ an optimal augmentation includes only the edge $\{v_0,v_{2k}\}$.

Note that the $(\frac{k}{2}-1)$-neighborhood of $v_k$ is the same in both $G_1$ and $G_2$, so it cannot distinguish between them in any algorithm that takes at most $\frac{k}{2}-1$ rounds. Hence, $v_k$ must have the same output in both cases. However, in $G_1$, both of the edges $\{v_{k-2},v_k\}$, $\{v_k,v_{k+2}\}$ are included in an optimal augmentation, and in $G_2$ they are not, so any distributed algorithm that solves TAP \emph{exactly} must take $\Omega(\frac{k}{2}-1)=\Omega(D)$ rounds. 

This lower bound holds also for approximation algorithms for the weighted problem: give the weight $1$ to the edge $\{v_0,v_{2k}\}$ and the weight $\alpha + 1$ to the edges $\{v_{2i},v_{2(i+1)}\}$ for $0 \leq i < k$. Any algorithm that adds at least one of the edges $\{v_{2i},v_{2(i+1)}\}$ to the augmentation has weight at least $\alpha + 1$, and hence is not an $\alpha$-approximation to weighted TAP. Therefore, any distributed $\alpha$-approximation algorithm for weighted TAP must take $\Omega(D)$ rounds.

A similar proof shows that approximating unweighted TAP takes $\Omega(D)$ rounds for appropriate values of $\alpha$. In the unweighted case, an algorithm that adds all the edges $\{v_{2i},v_{2(i+1)}\}$ for $0 \leq i < k$, gives a $k$-approximation to the optimal augmentation in $G_1$. However, if we want a better approximation we need $\Omega(D)$ rounds. Assume that $c>1$ is a constant and we want an $\alpha$-approximation where $\alpha < \frac{k}{c}=\frac{n-1}{2c}$. Consider the $\left \lceil \frac{k}{c} \right \rceil$ edges $\{v_{2i},v_{2(i+1)}\}$ that are closest to $v_k$. Each of the vertices on these edges is at distance $\Omega(k)=\Omega(D)$ from the vertices $v_0,v_{2k}$. Hence, they cannot distinguish between $G_1,G_2$ in less than $\Omega(D)$ rounds. It follows that any distributed $\alpha$-approximation algorithm for unweighted TAP must take $\Omega(D)$ rounds.
\end{proof}

\subsection{A Lower Bound for weighted TAP in the CONGEST model} \label{sec:app_congest}

By Theorem \ref{local-lb}, when $h=O(D)$ our algorithms {\unTAP}, {\weTAP} are optimal up to a constant factor. But what about the case of $h=\omega(D)$ for the CONGEST model? We next show a family of graphs where $h=\omega(D)$, in which $\Omega(h)$ rounds are needed in order to approximate weighted TAP, where $h=O(\sqrt{n})$.
The lower bound is proven using a reduction from the 2-party set-disjointness problem, in which there are two players, Alice and Bob. Each player gets a binary input string of length $k$: $a=(a_1,...,a_k), b=(b_1,...,b_k)$, and the players have to decide whether their inputs are disjoint, i.e., whether there is an index $i$ such that $a_i = b_i = 1$ or not. It is known that in order to solve this problem, Alice and Bob have to exchange at least $\Omega(k)$ bits, even when using random protocols \cite{razborov1992distributional}. 
Our construction is based on a construction presented in \cite{sarma2012distributed, elkin2006unconditional}. 
In order to use this construction for showing lower bounds for TAP, we add to it additional parallel edges\footnote{We also show a construction with no parallel edges.} and give weights to the edges in such a way that all the edges of the input tree $T$ can be covered by parallel edges of weight 0, except for $k$ edges, $\{e_i\}_{i=1}^k$. The edge $e_i$ may be covered either by a corresponding parallel edge $e_i^A$, or by a distant edge $e_i^B$ that closes a cycle that contains $e_i$. 
However, the weights of the edges $e_i^A$ and $e_i^B$ depend on the $i$'th bit in the input strings of Alice and Bob, such that there is a light edge that covers $e_i$ if and only if this bit equals 0 at least in one of the input strings. It follows that all the $k$ edges can be covered by light edges if and only if the input strings of Alice and Bob are disjoint. 

We next describe the construction.
We start by presenting a construction that includes parallel edges, and later explain how to change it to a similar construction that does not include parallel edges.
 
\subsubsection{Construction with Parallel Edges}

We follow the constructions presented in \cite{sarma2012distributed, elkin2006unconditional}.
Let $G_1=G(k,d,p)$ be a graph that consists of $k$ paths $P_1,...,P_k$ of length $d^p$, where the vertices on the path $P_i$ are denoted by $v^i_j$, for $0 \leq j \leq d^p-1$, and a tree $S$ of depth $p$, where each internal vertex has degree $d$, so it has $d^p$ leaves denoted by $u_j$, for $0 \leq j \leq d^p-1$. In addition, there is an edge between $u_j$ to $v^i_j$ for $1 \leq i \leq k$, $0 \leq j \leq d^p-1$. 

Let $G_2$ be a weighted graph with the same structure as $G_1$, and with parallel edges on the paths and in the tree. That is, there are two parallel edges between $v^i_j$ and $v^i_{j+1}$, for $0 \leq j<d^p-1$, and there are two parallel edges between a parent to each one of its $d$ children in $S$. All of the above parallel edges have weight $0$. In addition, there are two parallel edges between $u_0$ to $v^i_0$, one of them with weight $0$. The edges between $u_j$ to $v^i_j$ for $0 < j < d^p-1$ have weight $x=\alpha k+1$. Given two binary input strings of length $k$: $a=(a_1,...,a_k), b=(b_1,...,b_k)$, the second edge between $u_0$ and $v^i_0$ has weight $x$ if $a_i=1$ and has weight $1$ otherwise. Similarly, the edge between $u_{d^p-1}$ and $v^i_{d^p-1}$ has weight $x$ if $b_i=1$ and has weight $1$ otherwise. 

The input to the TAP problem is the graph $G_2$ with a spanning tree $T_{G_2}$ rooted at $r=u_0$ (see Figure \ref{pic2}). $T_{G_2}$ includes one copy of all the path edges, and one copy of all the edges of $S$, and the edges between $r=u_0$ and $v^i_0$ that have weight $0$ for $1 \leq i \leq k$. Since we can cover all the path edges and the edges of $S$ by their parallel edges having weight $0$, in order to cover all tree edges in $T_{G_2}$ optimally we need to cover the edges between $r$ and $v^i_0$ optimally. 

\setlength{\intextsep}{0pt}
\begin{figure}[h]
\centering
\setlength{\abovecaptionskip}{-10pt}
\setlength{\belowcaptionskip}{2pt}
\includegraphics[scale=0.35]{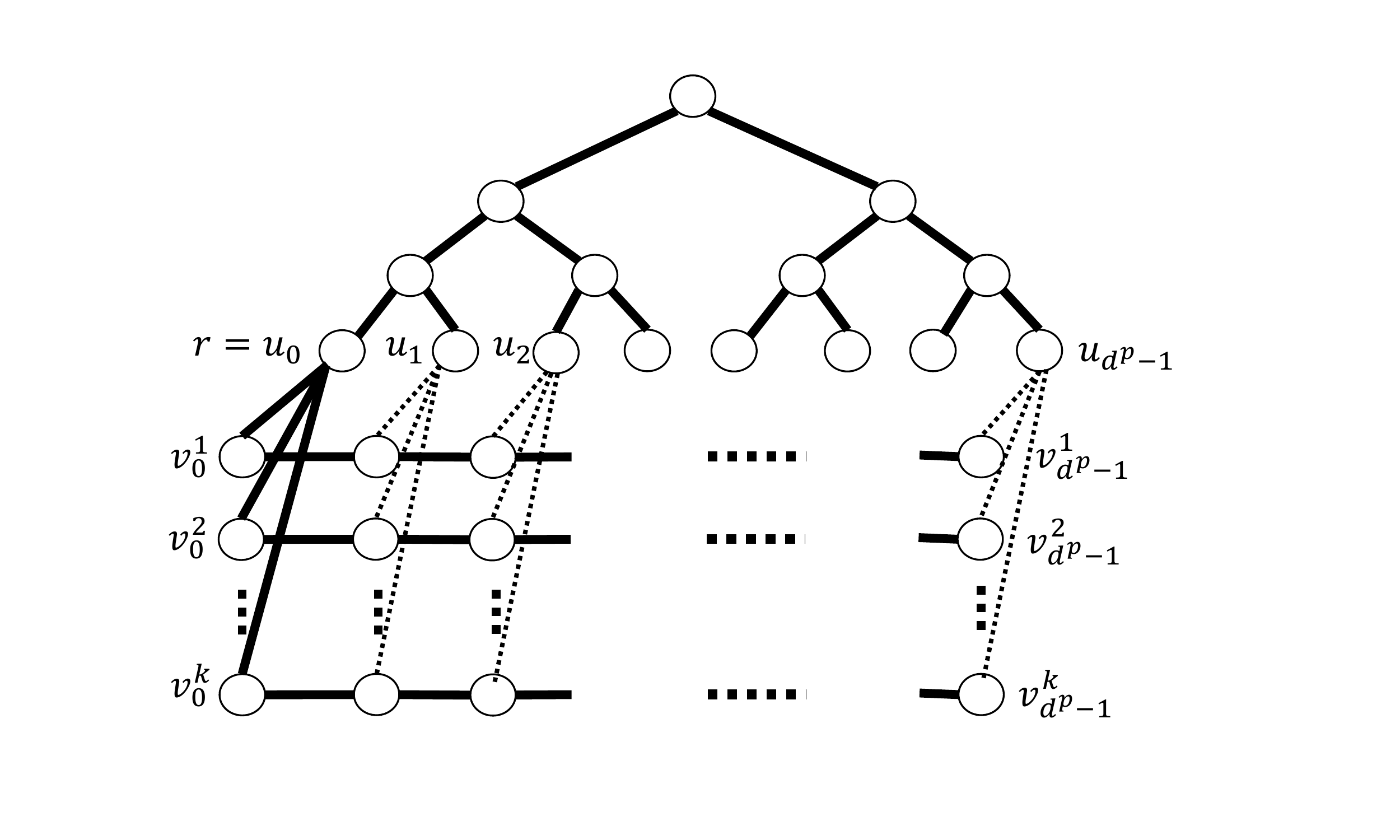}
 \caption{The structure of the graph $G_2$. The edges of $T_{G_2}$ are marked with solid lines, other edges are marked with dashed lines.} \label{pic2} 
\end{figure}

\begin{claim} \label{disj}
The cost of an optimal augmentation is $k$ if the input strings $a$ and $b$ are disjoint, and it is at least $x=\alpha k+1$ otherwise.
\end{claim}
\begin{proof}
In order to cover the tree edge $\{r,v^i_0\}$ we can use any other edge between $u_j$ to $v^i_j$. Each such edge has weight $x$ unless at least one of $a_i$ or $b_i$ is equal to $0$, in which case the second edge between $r=u_0$ to $v^i_0$ or the edge between $u_{d^p-1}$ and $v^i_{d^p-1}$ has weight $1$. These are the only edges that cover the tree edge $\{r,v^i_0\}$. All the other edges in $T_{G_2}$ can be covered with parallel edges of weight $0$. It follows that if $a$ and $b$ are disjoint then we can cover all the edges in $T_{G_2}$ with cost $k$, otherwise the cost is at least $x$ because we need at least one edge of weight $x$. 
\end{proof}

By Claim \ref{disj}, an $\alpha$-approximation algorithm that computes the weight of an optimal augmentation on the graph $G_2$ with spanning tree $T_{G_2}$ can be used in order to solve the set-disjointness problem: if the input strings are disjoint the weight of an optimal augmentation is $k$, in which case the output of the algorithm is at most $\alpha k$. Otherwise, the output of the algorithm is at least $x=\alpha k + 1$. 

Note that if $A$ is a distributed $\alpha$-approximation algorithm for weighted TAP that takes $R$ rounds, then there is a distributed $\alpha$-approximation algorithm $A_1$ for computing the weight of the optimal augmentation that completes in $O(R+D)$ rounds, where at the end of $A_1$ all the vertices know the weight of an optimal augmentation. This done by having $A_1$ simulate $A$ and then collect the weight of the augmentation over a BFS tree and distribute it to all the vertices. Since $R=\Omega(D)$ by Theorem \ref{local-lb}, it follows that the time complexity of $A_1$ is $O(R)$ rounds, so a lower bound on the time complexity of $A_1$ gives a lower bound on the time complexity of $A$.

Our algorithms work in the CONGEST model where the maximal message size is bounded by $\Theta(\log{n})$ bits, however the proof of the lower bound is based on the proof in \cite{sarma2012distributed} which works in a more general model where the maximal message size is bounded by $B$ bits. Hence, the lower bound we show holds for this generalized model as well.

\begin{claim} \label{sim}
If there is a distributed (even randomized) $\alpha$-approximation algorithm for computing the weight of an optimal augmentation in $G_2$ that has time complexity of $R$ rounds where $R \leq \frac{d^p-1}{2}$, then set-disjointness can be solved by exchanging $O(dpBR)$ bits. 
\end{claim}

\begin{proof}
The proof of the claim follows from the proof of Theorem 3.1 in \cite{sarma2012distributed}, in which it is shown how Alice and Bob can simulate a distributed algorithm on the graph $G_1$ by exchanging at most $2dpBR$ bits, where at the end of the simulation each player knows the output of one of the vertices $r, u_{d^p-1}$. In the algorithm for computing the weight of an optimal augmentation all the vertices know the weight at the end, so it is enough that each of Alice and Bob knows the output of one vertex. Note that the graphs $G_1$ and $G_2$ have the same structure, but in $G_2$ there may be two parallel edges between vertices $v,u$ that have only one edge between them in $G_1$. It follows that $v,u$ can exchange $2B$ bits between them in a round in each direction, instead of $B$ bits. Therefore, in order to simulate a distributed algorithm on $G_2$, Alice and Bob can use the same simulation but may need to exchange twice as many bits in order to simulate one round, and $4dpBR$ bits for the whole simulation, which is still $O(dpBR)$ bits, as claimed. At the end of the simulation, both Alice and Bob know an $\alpha$-approximation to the weight of an optimal augmentation, and can deduce if their input strings are disjoint according to Claim \ref{disj}.  
\end{proof}

\begin{theorem} (equivalent to Theorem 7.1 in \cite{sarma2012distributed}) \label{lowerbound2}
For any polynomial function $\alpha(n)$, integers $p > 1$, $B \geq 1$ and $n \in \{2^{2p+1}pB, 3^{2p+1}pB,...\}$, there is a $\Theta(n)$-vertex graph of diameter $2p+2$ for which any (even randomized) distributed $\alpha(n)$-approximation algorithm for weighted TAP with an instance tree $T \subseteq G$ of height $h$ requires $\Omega((n/(pB))^{\frac{1}{2}-\frac{1}{2(2p+1)}})$ rounds which is $\Omega(h)$.
\end{theorem}

\begin{proof}
By Claim \ref{sim} and the lower bound on set-disjointness \cite{razborov1992distributional} we have $R=\Omega(min(d^p,\frac{k}{dpB}))$. Choosing $k=d^{p+1}pB$ gives $\Omega(min(d^p,\frac{k}{dpB}))=\Omega(d^p)$. As in \cite{sarma2012distributed,elkin2006unconditional}, $G_1$ and $G_2$ have $n=\Theta(kd^p)=\Theta(d^{2p+1}pB)$ vertices and diameter $2p+2$. In addition, $h = d^p + 1$ since the height of $T_{G_2}$ is determined by the length of the paths. Hence, we have $R = \Omega(d^p) = \Omega(h)$ where $h = \Theta(d^p) = \Theta((n/(pB))^{\frac{1}{2}-\frac{1}{2(2p+1)}})$.
\end{proof}

Choosing $B=p=\Theta(\log{n})$ in Theorem \ref{lowerbound2} gives the following.

\congest*

\subsubsection{Construction without Parallel Edges}

We next explain how to modify the above construction to avoid parallel edges. 
We define $G_3$ as follows: 
\begin{itemize}
\item If there is a single edge between the vertices $v$ and $u$ in $G_2$, then this edge is in $G_3$ and has the same weight as it has in $G_2$.
\item For every pair of vertices $v,u$ which have two parallel edges between them in $G_2$, we add in $G_3$ a new vertex $vu$ and replace one of the two parallel edges between $v$ and $u$ which has weight $0$ by two edges $\{v,vu\}$ and $\{vu,u\}$, both with weight $0$.\footnote{Notice that at least one of the two parallel edges indeed has weight $0$.}
\end{itemize}

The tree $T_{G_3}$ in the TAP problem in $G_3$ is constructed according to the tree $T_{G_2}$ in $G_2$, such that if $\{v,u\}$ is a tree edge in $T_{G_2}$, then $\{v,vu\},\{vu,u\}$ are tree edges in $T_{G_3}$. Note that the edge $\{v,u\}$ covers both $\{v,vu\},\{vu,u\}$. Since all the edges on the paths and in the tree $S$ in $G_2$ have weight 0, all the edges on the corresponding paths and tree $S_{G_3}$ in $T_{G_3}$ can be covered by edges of weight $0$. In order to cover all tree edges in $T_{G_3}$ optimally we need to cover the edges $\{r,rv^i_0\},\{rv^i_0,v^i_0\}$ optimally. Similarly to the case in $G_2$, we can cover them by any one of the edges between $u_j$ and $v^i_j$. All those edges have weight $x$ unless at least one of $a_i$ or $b_i$ is equal to $0$, so Claim \ref{disj} holds for $G_3$ as well.

If $n$ is the number of vertices in $G_2$, then in $G_3$ the number of vertices is $2n-1=\Theta(n)$ because we add one vertex for each edge of $T_{G_2}$ (the parallel edges in $G_2$ are only on the tree $T_{G_2}$). Similarly, the height of $T_{G_3}$ is $2h=\Theta(h)$ where $h$ is the height of $T_{G_2}$, and the diameter of $G_3$ is $\Theta(D)$ where $D$ is the diameter of $G_2$.

Assume that $A$ is an $\alpha$-approximation algorithm for weighted TAP that takes $R$ rounds in $G_3$, then there is an $\alpha$-approximation algorithm $A_1$ for weighted TAP that takes $R$ rounds in $G_2$. $A_1$ simulates $A$: all the vertices that are both in $G_2$ and in $G_3$ simulate themselves. For each vertex $vu$, one of the vertices $v,u$ simulates $vu$, and assume w.l.o.g that $v$ simulates $vu$. Note that there are two parallel edges between $v$ and $u$ in $G_2$. One of them is used in order to simulate the messages sent on the edge $\{v,u\}$ in $A$, and the other is used in order to simulate the messages sent on the edge $\{vu,u\}$ in $A$. Note that there is no need for communication in order to simulate messages sent on the edge $\{v,vu\}$ because the vertex $v$ simulates both $v,vu$. It follows that the simulation of $A$ in $G_2$ takes $R$ rounds. In addition, from the correspondence between $G_2$ and $G_3$, any augmentation in $G_3$ is an augmentation in $G_2$, and vice versa. 

The above implies that the lower bound holds for $G_3$ (which has no parallel edges) as well, and hence Theorem \ref{lowerbound2} holds also for simple graphs.

\section{Discussion} \label{sec:dis}

In this paper, we present the first distributed approximation algorithms for TAP.  Many intriguing problems remain open. First, can  we get efficient distributed algorithms for TAP with an approximation ratio better than 2? In the sequential setting, achieving an approximation better than 2 for weighted TAP is a central open question. However, there are several recent algorithms achieving better approximations
for unweighted TAP \cite{kortsarz2016simplified, cheriyan2015approximating, DBLP:conf/stoc/0001KZ18} or for weighted TAP with bounded weights \cite{DBLP:journals/corr/FioriniGKS17, adjiashvili2017beating}.


Second, there are many additional connectivity augmentation problems, such as increasing the edge connectivity from $k$ to $k+1$ or to some function $f(k)$, as well as augmentation for increasing the vertex connectivity. Such problems have been widely studied in the sequential setting, and a natural question is to design distributed algorithms for them.

Finally, it is interesting to study TAP and additional connectivity problems also in other distributed models, such as the dynamic model where edges or vertices may be added or removed from the network during
the algorithm. An interesting question is how to maintain highly-connected backbones when the network can change dynamically.

\bibliographystyle{spmpsci}
\bibliography{TAP}

\end{document}